\documentclass[aos,preprint]{imsart}

\pdfoutput=1                    
\setattribute{journal}{name}{}   

\usepackage[OT1]{fontenc}
\usepackage{amsthm,amsmath,graphicx,enumerate}
\usepackage{natbib}
\usepackage{amssymb,tabularx,multicol,multirow,booktabs}
\usepackage{mhequ}
\usepackage{relsize}
\usepackage[title,toc,page]{appendix}
\usepackage{xr-hyper} 
\usepackage[colorlinks,citecolor=blue,urlcolor=blue,filecolor=blue]{hyperref}  
\usepackage{color,soul,subfigure}
\usepackage{rotating,mathrsfs}     

\usepackage{float}
\restylefloat{table}
\restylefloat{figure}

\arxiv{arXiv:1701.05230}

\startlocaldefs
\numberwithin{equation}{section}
\theoremstyle{plain}

\endlocaldefs

\newtheorem{theorem}[]{Theorem}[section]
\newtheorem{lemma}[]{Lemma}[section]
\newtheorem{assumption}[]{Assumption}[section]
\newtheorem{remark}[]{Remark}[section]

\def\bx{\mathbf{x}}
\def\bX{\mathbf{X}}

\def\bZ{\mathbf{Z}}
\def\bu{\mathbf{u}}

\def\bv{\mathbf{v}}

\def\bt{\mathbf{t}}
\def\bT{\mathbf{T}}

\def\E {\mathbb{E}}
\def\G {\mathbb{G}}

\def\P {\mathbb{P}}
\def\T {\mathbb{T}}

\def\bbeta{\boldsymbol{\beta}}
\def\balpha{\boldsymbol{\alpha}} 
\def\btheta{\boldsymbol{\theta}}
\def\beps{\boldsymbol{\epsilon}}

\def\bSigma{\boldsymbol{\Sigma}}
\def\bxi{\boldsymbol{\xi}}

\def\bbetahat{\widehat{\bbeta}}

\newcommand\ind{\protect\mathpalette{\protect\independenT}{\perp}}
\def\independenT#1#2{\mathrel{\rlap{$#1#2$}\mkern4mu{#1#2}}}

\newcommand\given[1][]{\:#1\vert\:}

\def\Lsc{\mathcal{L}}
\def\Usc{\mathcal{U}}
\def\Isc{\mathcal{I}}
\def\Xsc{\mathcal{X}}

\def\Ztil {\widetilde{\mathbb{Z}}}

\def\Var{\mbox{Var}}

\def\half{\frac{1}{2}}

\def\bone{\mathbf{1}}
\def\bzero{\mathbf{0}}

\def\Msc{\mathcal{M}}

\def\bT{\mathbf{T}}

\def\S{\mathbb{S}}

\def\deltabar{\overline{\delta}}

\def\nq{n_q}
\def\bgamma{\boldsymbol{\gamma}}
\def\L{\mathbb{L}}
\def\bmu{\boldsymbol{\mu}}
\def\bbetabar{\overline{\bbeta}}
\def\balphabar{\overline{\balpha}}
\def\bXbar{\overline{\bX}}
\def\Ybar{\overline{Y}}
\def\C{\mathbb{C}}
\def\Asc{\mathcal{A}}
\def\bvj{\bv_{[j]}}
\def\R{\mathbb{R}}
\def\Jsc{\mathcal{J}}
\def\Msc{\mathcal{M}}
\def\T{\mathbb{T}}
\def\Tnqo{\T_{\nq}^{(1)}}
\def\Tnqto{\T_{\nq, 1}^{(2)}}
\def\Tnqtt{\T_{\nq, 2}^{(2)}}
\def\dbar{\overline{d}}
\def\Pibac{\Pi^c_{\bbeta_0}(\balpha_0)}
\def\S{\mathbb{S}}
\def\G{\mathbb{G}}
\def\SGq{\S\G_q}
\def\bXtil{\widetilde{\bX}}
\def\Ytil{\widetilde{Y}}
\def\pqtil{\widetilde{p}_q}
\def\piqtil{\widetilde{\pi}_q}
\def\Nsc{\hbox{Normal}}
\def\zbar{\overline{z}}
\def\rhotil{\widetilde{\rho}}
\def\bt{\mathbf{t}}
\def\Z{\mathbb{Z}}
\def\Znq{\mathcal{Z}_{n_q}}
\def\Cmin{C_{\min}}
\def\Cmax{C_{\max}}
\def\Phibar{\overline{\Phi}}
\def\abar{\overline{a}}
\def\bbar{\overline{b}}
\def\cbar{\overline{c}}
\def\Rsc{\mathcal{R}}
\def\Ztil{\widetilde{\Z}}
\def\bxi{\boldsymbol{\xi}}
\def\bbetatil{\widetilde{\bbeta}}

\def\bD{\mathbf{D}}
\def\Lsc{ {\mathcal L}}

\def\argmin{\mbox{argmin}}
\def\bSigma{\boldsymbol{\Sigma}}

\def\hs2{\hspace{2mm}}

\definecolor{jcolor}{RGB}{041,122,000}
\definecolor{darkred}{RGB}{100,000,000}
\definecolor{purple}{RGB}{200,000,200}

\def\boxit#1{\vbox{\hrule\hbox{\vrule\kern6pt  \vbox{\kern6pt#1\kern6pt}\kern6pt\vrule}\hrule}}

\def\ULASSO{\mbox{U}_{\mbox{\tiny LASSO}}}

\def\Fscr{\mathscr{F}}

\def\bT{\mathbf{T}}

\def\hbar{\overline{h}}

\def\bmu{\boldsymbol{\mu}}

\def\bc{\mathbf{c}}

\begin{document}

\begin{frontmatter}

\title{Surrogate Aided Unsupervised Recovery of Sparse Signals in Single Index Models for Binary Outcomes} 
\runtitle{Surrogate Aided Unsupervised Recovery of Sparse Signals}

\begin{aug}
 \author{\fnms{Abhishek} \snm{Chakrabortty}\corref{}\ead[label=e1]{abhich@mail.med.upenn.edu}\thanksref{t1,t5},}
\author{\fnms{Matey} \snm{Neykov}\corref{}\ead[label=e2]{mneykov@andrew.cmu.edu}}, 
\author{\fnms{Raymond} \snm{Carroll}\corref{}\ead[label=e3]{carroll@stat.tamu.edu}} 
\and \author{\fnms{Tianxi} \snm{Cai}\corref{}\ead[label=e4]{tcai@hsph.harvard.edu}\thanksref{t5}}

\thankstext{t1}{Corresponding author; previously at Harvard University during the time of this work.} 
\thankstext{t5}{This research was partially supported by the National Institutes of Health grants U54 HG007963 and U01 CA057030.}

\affiliation{University of Pennsylvania, Carnegie Mellon University, Texas A\&M University and Harvard University}

\runauthor{A. Chakrabortty et al.}

\address{Abhishek Chakrabortty\\
Department of Statistics\\
University of Pennsylvania\\
Philadelphia, PA 19104, USA.\\
\printead{e1}\\
\phantom{E-mail:\ }}

\address{Matey Neykov\\
Department of Statistics \& Data Science\\
Carnegie Mellon University\\
Pittsburgh, PA 15213, USA.\\
\printead{e2}\\
\phantom{E-mail:\ }}

\address{Raymond Carroll \\
Department of Statistics\\
Texas A\&M University\\
College Station, TX 77843-3143, USA.\\
\printead{e3}\\
\phantom{E-mail:\ }}

\address{Tianxi Cai\\
Department of Biostatistics\\
Harvard University\\
Boston, MA 02115, USA.\\
\printead{e4}\\
\phantom{E-mail:\ }}
\end{aug}



\begin{abstract}

We consider the recovery of regression coefficients, denoted by $\bbeta_0$, for a single index model (SIM) relating a binary outcome $Y$ to a set of possibly high dimensional covariates $\bX$, based on a large but \emph{unlabeled} dataset $\Usc$, with $Y$ never observed. On $\Usc$, we fully observe $\bX$ and additionally, a surrogate $S$ which, while not being strongly predictive of $Y$ throughout the entirety of its support, can forecast it with high accuracy when it assumes extreme values. Such datasets arise naturally in modern studies involving large databases such as electronic medical records (EMR) where $Y$, unlike $(\bX, S)$, is difficult and/or expensive to obtain. In EMR studies, an example of $Y$ and $S$ would be the true disease phenotype and the count of the associated diagnostic codes respectively. Assuming another SIM for $S$ given $\bX$, we show that under sparsity assumptions, we can recover $\bbeta_0$ proportionally by simply fitting a least squares LASSO estimator to the subset of the observed data on $(\bX, S)$ restricted to the extreme sets of $S$, with $Y$ imputed using the surrogacy of $S$. We obtain sharp finite sample performance bounds for our estimator, including deterministic deviation bounds and probabilistic guarantees. We demonstrate the effectiveness of our approach through multiple simulation studies, as well as by application to real data from an EMR study conducted at the Partners HealthCare Systems.
\end{abstract} %

\end{frontmatter}

\vbadness10000 

\section{Introduction}\label{P2:intro}

\subsection{Background}

Unsupervised classification methods are of great importance in a wide variety of scientific applications including image retrieval and processing, document classification, genome-phenome association analysis and other problems in biomedical sciences \citep{gllavata2004text, chen2005clue, henegar2006unsupervised}.
In recent years, many unsupervised  learning methods have been proposed to classify categorical outcomes. Examples include clustering, latent class mixture modeling, neural networks and random forest based methods \citep{merkl2000document, hofmann2001unsupervised, shi2006unsupervised, cios2007unsupervised, kosorok_2013}. Most of the related existing literature, however, largely focuses on identifying algorithms that can accurately classify the outcomes of interest with less focus on the statistical properties of the estimated model parameters. We consider in this paper 
a surrogate aided unsupervised classification problem of a very \emph{different} and unique nature. Motivated by the problem of automated phenotyping with electronic medical records (EMR) data, among other  problems, we consider a regression modeling approach to unsupervised classification with assistance from a surrogate variable 
$S$ whose extreme (but \emph{not} all) values are highly predictive of an unobserved binary outcome (or label) $Y \in \{0,1\}$.

Specifically, we consider relating  $Y$ to a $p\times 1$ covariate vector $\bX$ through a flexible single index model (SIM) under which the regression parameter, 
$\bbeta_0$, is identifiable only up to scalar multiples. The available data, $\Usc$, is large in size, but completely unlabeled with $Y$ \emph{never} observed (our usage of the term `unsupervised' is meant precisely in this sense). 
On $\Usc$, we observe $\bX$, and $S$ which is \emph{not} necessarily strongly predictive of $Y$ throughout the entirety of its support, but can 
forecast $Y$ with reasonable accuracy when it assumes extreme values. Such data arises naturally in settings where $Y$, unlike $\bX$ and $S$, is difficult and/or expensive to obtain, a scenario that is of great practical relevance especially in the modern `big data' era with massive unlabeled datasets becoming increasingly available and tractable. In particular, they are frequently encountered in modern biomedical studies involving analyses of large databases like EMR, where an example of $Y$ and $S$ could be a disease phenotype for conditions like rheumatoid arthritis (RA) and the count of International Classification of Diseases, Ninth Revision (ICD9) codes for RA, respectively. We first briefly discuss the motivating problem of automated phenotyping in 
EMR, followed by a summary of our contributions in this paper and the proposed framework for unsupervised recovery of sparse signals in single index models for binary outcomes using extremes of a surrogate variable.

\subsection{A Motivating Example (Automated Phenotyping in EMR)}\label{P2:intro_example}

EMR linked with bio-repositories provide rich resources of data for discovery research \citep{Kohane_2011}.
Integrative analyses of large scale clinical and phenotypic data, readily available from the EMR, and biological data from bio-repositories can be performed to rigorously study genome-phenome association networks and improve the understanding of disease processes and treatment responses \citep{wilke2011emerging,kohane2012translational}. For example, when new genetic variants are discovered, the scope of their clinical significance can be assessed by examining the range of disease phenotypes that are associated with these variants via phenome-wide association studies (PheWAS) based on EMR cohorts \citep{denny2010phewas}. EMR data are the key to the success of PheWAS as they contain nearly complete clinical diagnoses, broadening the ability to simultaneously test for potential associations between genetic variants and a wide range of disorders, in contrast to traditional cohort studies that typically focus 
on only a few predetermined disease phenotypes as outcomes.

However, despite its potential for translational research, one major rate-limiting step in EMR driven PheWAS is the difficulty in extracting accurate information on the true disease phenotype $Y$ from the EMR, which usually requires labor intensive manual chart review by physicians \citep{bielinski2011mayo}. Current PheWAS methods primarily rely on ICD9 codes to assess the phenotype \citep{denny2010phewas, Liao_2010}. The ICD9 codes have limited predictive accuracy for many diseases and hence, can introduce substantial noise into the subsequent association studies. For example, based on data from the Partners HealthCare Systems, among subjects with at least 3 RA ICD9 codes, only 56\% of those actually have confirmed RA \citep{Liao_2010}. However, for the subsets of patients with very high or low counts of RA ICD9 codes, the ICD9 codes can predict the true RA status with a high degree of accuracy, thereby serving as an effective surrogate outcome in these subsets (see Section \ref{P2:data_ex}). Appropriate and efficient use of such available surrogacy information can lead to unsupervised algorithms that can accurately predict $Y$ and hence, significantly reduce the burden of manual labeling in large cohorts.
In particular for EMR data, such automated unsupervised classification algorithms can pave the way for high throughput phenotyping \citep{Murphy_2009, ritchie2010robust, Yu_2015}, allowing for phenome-genome association studies that typically requires the availability of multiple phenotypes and hence does not scale well with manual labeling methods for obtaining gold standard labels/outcomes.

Finally, it is worth mentioning that the example with EMR data discussed here serves only as \emph{one} main motivation for the statistical problem we consider in this paper. Our framework (and our proposed methodology) indeed applies more generally to several other interesting problems that are of particular relevance in modern studies involving large databases, where unlabeled data for $\bX$, as well as observations for a suitable surrogate $S$, may be available in plenty but the observations for the corresponding binary outcome (or label) $Y$ may be difficult to obtain (possibly due to logistics, cost issues etc., among other reasons). Such settings, in general, are of considerable interest in unsupervised as well as semi-supervised learning \citep{Chapelle_2006} and some classic examples from machine learning include text mining, web page classification, speech recognition, natural language processing etc.

\subsection{Contributions of this Paper}\label{P2:contrib}

In this paper, we propose an unsupervised estimator of $\bbeta_0$ by making use of the extremes of a surrogate variable $S$ that is observable for the entire population.  Specifically, under another SIM for $S$ given $\bX$ with an unknown parameter $\balpha_0$, we propose to estimate $\bbeta_0$ by regressing a surrogate binary outcome $Y^*$, defined by $S$, on $\bX$ via a simple $L_1$-penalized linear regression in an extreme subset of $S$ consisting of $100q\%$ of the study population, for some small $q \in (0,1]$. Under sparsity assumptions on $\bbeta_0$ and conditions controlling the misclassification error, denoted by $\pi_q$, of $Y^* \equiv Y^*_q$ for $Y$ in the extreme subset, we show that with $p$ possibly large, our proposed simple Unsupervised LASSO ($\ULASSO$) estimator recovers $\bbeta_0$ up to a scale multiplier.
We also obtain explicit finite sample (and deterministic) deviation bounds for the performance of our estimator, along with high probabilistic guarantees for the bounds to obey satisfactory convergence rates. The results have several useful implications, including an interesting `variance-bias  tradeoff' (in terms of $q$) in the convergence rates, whereby for a given order of the misclassification error $\pi_q$, the corresponding optimal order of $q$ and the optimal convergence rate can also be determined. We also explicitly characterize the behaviour of $\pi_q$ versus $q$ for one specific setting, wherein the interplay between $\bbeta_0$ and $\balpha_0$ and the necessary conditions for our approach to succeed become more explicit.

The rest of this paper is organized as follows. In Section \ref{P2:psetup}, we detail the problem set-up, as well as some key ideas and results motivating our approach. In Section \ref{P2:ULASSO}, we present the $\ULASSO$ estimator and all its theoretical properties. Results from extensive simulation studies are given in Section \ref{P2:sim}. 
The performance of the $\ULASSO$ estimator is found to be comparable, and in fact \emph{better} in most cases, to that of a supervised estimator based on as many as $500$ gold-standard labels. Our estimator also does not appear to be too sensitive to the choice of $q$ provided that it is small enough. In Section \ref{P2:data_ex}, we apply our method to an EMR study in which a labeled set of observations is also available for validation. The results indicate that our estimator works well in real applications. Finally, concluding discussions are given in Section \ref{P2:disc}. Technical materials, including assumptions, supporting lemmas and detailed proofs of all our theoretical results, are distributed in Appendices \ref{P2:assumption}-\ref{P2:thm3.2_proof}. 
Further, in Appendix \ref{P2:splcase}, 
we also provide some results, along with useful discussions, to illustrate how our approach and main results can be applied to a specific subclass of 
models for $(Y,S, \bX')'$ that are of considerable interest in the literature and are frequently adopted in practice.

\section{Problem Set-Up}\label{P2:psetup}

\emph{Notations and Assumptions.}\label{P2:notn_assmpn}
We first introduce some notations to be used throughout. For any $\bv \in \mathbb{R}^p$ and $j \in \{1,\hdots, p\}$, let $\bvj$ denote the $j^{th}$ coordinate of $\bv$, $\| \bv \|_r$ the $L_r$ norm of $\bv$ $ \; \forall \; r \geq 0$, $\Asc(\bv) \equiv \{ j: \bvj \neq 0 \}$ the support of $\bv$, and $s_{\bv} \equiv \| \bv \|_0 $ the cardinality of $\Asc(\bv)$. Further, for any $\Jsc \subseteq \{1, \hdots,p\}$, let $\Jsc^c = \{1, \hdots, p\} \backslash \Jsc$, $\Msc_{\Jsc} = \{ \bv \in \R^p : \Asc(\bv) \subseteq \Jsc \}$ and $\Msc_{\Jsc}^{\perp} = \{ \bv \in \R^p : \Asc(\bv) \subseteq \Jsc^c \}$, and let $\Pi_{\Jsc}(\bv)$ denote the $p\times 1$ vector with $j$th element being $\bv_{[j]}1\{j \in \Jsc\}$. We use the shorthand $\Pi_{\bv}(\cdot)$ and $\Pi_{\bv}^c(\cdot)$ to denote $\Pi_{\Asc(\bv)}(\cdot)$ and $\Pi_{\Asc^c(\bv)}(\cdot)$ respectively. 
Next, for any positive definite (p.d.) matrix $\bSigma$, denoted as $\bSigma \succ 0$, we let $\lambda_{\min}(\bSigma) > 0$ denote the minimum eigenvalue of $\bSigma \succ 0$. Lastly, for any $d \geq 1$, we denote by $\Nsc_d(\bmu,\bSigma)$ the $d$-variate Gaussian distribution with mean $\bmu \in \R^d$ and covariance matrix $\bSigma_{d \times d} \succ 0$, by $\mbox{Logistic}(a,b)$ the logistic distribution with mean $a \in \R$ and variance $b > 0$, and by $\mbox{Uniform(a,b)}$ the uniform distribution on $(a,b)$ for any $a, b \in \R$ with $a < b$.

We assume throughout that $\bD = (Y, S, \bX')'$ is defined on a common probability space with probability measure $\P(\cdot)$ and has finite 2\textsuperscript{nd} moments. Let $\E(\cdot)$ denote expectation with respect to $\P(\cdot)$. For any $q \in (0,1]$, let $\delta_q$ and $\deltabar_q$ respectively denote the $(q/2)^{th}$ and $(1-q/2)^{th}$ quantiles of 
$S$, and define: 
$$\Isc_q \; = \; (-\infty,\delta_q] \hspace{0.01in} \cup \hspace{0.01in} [\deltabar_q,\infty) \quad \forall \; q \in (0,1].$$
Let $\P_q(\cdot)$ denote the  probability measure characterizing the distribution of $\bD \given S \in \Isc_q$ and let
$\E_q(\cdot)$ denote expectation with respect to $\P_q(\cdot)$. Let $p_q = \E_q(Y)$, $\bmu_q = \E_q(\bX)$ and $\bSigma_q = \Var(\bX \given S \in \Isc_q)$, where we assume $\bSigma_q \succ 0$.
Finally, let $\pi_q^- = \P(Y=1 \given S\leq \delta_q)$ and $\pi_q^+ = \P(Y=0 \given S\geq \overline{\delta}_q)$. The premises of our problem, as
formalized in Assumption \ref{P2:surr_assmpn} later, entail that $\pi_q^-$ and $\pi_q^+$ are both small for small enough $q$.

The underlying data consists of $N$ independent and identically distributed (i.i.d.)  realizations of $\bD$, denoted as:
$\Fscr_N^* = \{\bD_i = (Y_i,S_i,\bX_i')': i = 1,\hdots,N\}$, while the \emph{observed data}, completely unlabeled, consists of the $N$ i.i.d. realizations of $(S,\bX')'$ only, denoted as:
$$\Usc_N^* = \{(S_i,\bX_i')': i = 1, \hdots,N\}.$$
The variable $S$, in very heuristic terms, satisfies the following property: it is known a priori, based on domain knowledge, that when $S$ is `too low' or `too high', then the corresponding $Y$ is `very likely' to be $0$ or $1$ respectively. We formalize this assumption as follows.
\begin{assumption}[Extreme Tail Surrogacy of $S$]\label{P2:surr_assmpn}
 \emph{Define a \emph{surrogate outcome} $Y^*_q  = 1(S\geq \deltabar_q)$ in the subset $S \in \Isc_q$ and let $\pi_q = \P_q(Y \neq Y^*_q)$, where $1(\cdot)$ denotes the indicator function.  We assume that for some universal constants $\nu, C > 0$, 
 and some $q_0 \in (0,1]$ small enough,
\begin{equation}
\pi_q \; \equiv \; \P_q(Y \neq Y^*_q) \; \equiv \; \half(\pi_q^- + \pi_q^+) \; \leq \; Cq^\nu \quad \forall \; q \leq q_0 \in (0,1]. 
\label{P2:misclass_error}
\end{equation} }
\end{assumption}
Note that the surrogacy in Assumption \ref{P2:surr_assmpn} is formulated in terms of the quantiles of $S$ and therefore allows for the support of $S$ to be of arbitrary nature (continuous and/or discrete).

For most of this paper, our primary focus will be on the subsets of $\Usc_N^*$ and $\Fscr_N^*$ consisting of the observations for which $S \in \Isc_q$, defined as follows:
\begin{eqnarray}
\Usc_{\nq} &=& \{(Y_{q,i}^*, S_i,\bX_i')': \; S_i \in \Isc_q, \; i = 1,\hdots, \nq \equiv Nq\}, \;\; \mbox{and} \label{P2:restricted_surr_data} \\
\Fscr_{\nq} &=& \{(Y_i,S_i,\bX_i')': \; S_i \in \Isc_q, \; i = 1,\hdots, \nq \equiv Nq\},\label{P2:restricted_data}
\end{eqnarray}
where without loss of generality (WLOG), we re-index the observations in both $\Usc_{\nq}$ and $\Fscr_{\nq}$ for notational ease.
The sample size $N$ is assumed to be substantially \emph{large} (see Remark \ref{P2:size_remark}), so that the distribution of $(S,\bX')'$ can be presumed to be (almost) known for all practical purposes. We shall hence assume for simplicity that $\delta_q$ and $\deltabar_q$ are known as well
\footnote{In practice, the (unknown) population quantities $\delta_q$ and $\deltabar_q$ may be (near-perfectly) estimated from the original observed data $\Usc_N^*$ whose size $N$ is very large and more importantly, since $q \rightarrow 0$, $N$ is of much higher order compared to the effective sample size, $n_q \equiv Nq$, of our eventual dataset of interest $\Usc_{\nq}$. In our theoretical formulations, we shall therefore ignore for simplicity this minor (and lower order) source of randomness involved in estimating the quantiles of $S$ from $\Usc_N^*$ and assume that they are known.}.
Lastly, while all results obtained in this paper for our proposed estimators are finite sample results, they are essentially derived with the following regime in mind: $N \rightarrow \infty$, $q = O(N^{-\eta})$ for some constant $\eta \in (0,1)$, so that $q \rightarrow 0$ and $n_q = O(N^{1-\eta}) \rightarrow \infty$, as $N \rightarrow \infty$.

\begin{remark}[Large Size of the Original Data $\Usc_N^*$]\label{P2:size_remark}
\emph{
The fact that $N$ is `very large' (for example, $N \asymp 10^5$) is key to the relevance of our problem and its premises, and to the potential success of our proposed approach based on $\Usc_{\nq}$. It ensures that even for small enough $q$ such as 
$q \asymp 10^{-2}$, so that (\ref{P2:misclass_error}) can be made to hold, our effective sample size $\nq$ $(\asymp 10^3)$ is still large enough to lead to an estimator with reasonable stability and convergence guarantees. More importantly, such choices of $(N, q)$ ensure that $\nq$ easily remains comparable to the \emph{maximum size} of a 
labeled data that one can realistically hope to procure in practice, given the logistic constraints 
typically involved in obtaining $Y$ under our settings of interest.
}
\end{remark}

\subsection{Model Assumptions}\label{P2:model_assumptions}
We assume throughout the following \emph{single index models} (SIMs) for $Y $ and $S$ given $\bX$.
\begin{eqnarray}\label{P2:sim_models}
Y & = & f(\bbeta_0'\bX;\epsilon)  \quad\; \mbox{with} \; \epsilon \ind (S,\bX) \; \mbox{and} \; f(\cdot)\; \mbox{\emph{unknown}}, \quad \mbox{and} \label{P2:y_model}\\
S & = & g(\balpha_0'\bX;\epsilon^*)  \quad \mbox{with} \; \epsilon^* \ind \bX, \; \epsilon^* \ind \epsilon, \; \mbox{and} \; g(\cdot) \; \mbox{\emph{unknown}}, \label{P2:s_model}
\end{eqnarray}
where $\bbeta_0, \balpha_0 \in \R^p$ are unknown parameter vectors, and $(\epsilon, \epsilon^*)$ represent the corresponding random noise components. Since $f(\cdot)$ and $g(\cdot)$ are unspecified, $\bbeta_0$ and $\balpha_0$ are identifiable \emph{only} up to scalar multiples; see Section \ref{P2:disc} for further discussions on the model assumptions.
Note that in (\ref{P2:s_model}), we do not require $S$ to be continuous (for instance, it can be a count variable, as in the example in Section \ref{P2:intro_example}). The map $g(\cdot)$ in (\ref{P2:s_model}) is, in general, $\Xsc_S$-valued, where $\Xsc_S \subseteq \R$ denotes the appropriate support of $S$. 
As for the map $f(\cdot)$ in (\ref{P2:y_model}), it can be viewed as: $f(\bbeta_0'\bX;\epsilon)$ $= 1\{\bar{f}(\bbeta_0'\bX;\epsilon) > 0\}$ for some unknown $\R$-valued function $\bar{f}(\cdot)$.
%

SIMs have been widely studied in classical econometrics \citep{Powell_SS_1989, Ichimura_1993, Horowitz_2009}, as well as in statistics as part of the sufficient dimension reduction literature \citep{Duan_Li_1989, Duan_1991, K-C_Li_1991, Cook_2009}. Analysis of SIMs in high dimensional settings has also garnered considerable attention in recent years in the relevant statistics literature \citep{Goldstein_2016, Neykov_L1_2016, Wei_2018} as well as in the compressed sensing literature  \citep{Thramp_2015, Plan_Vershynin_2013_a, Plan_2016, Plan_2017}. Several of these recent works build upon and/or are also closely related to the seminal results and insights of \citet{Duan_Li_1989} on `link free' regression which serve as a main inspiration of our approach as well; see Section \ref{P2:basic_found} for further details.

The models (\ref{P2:y_model}) and (\ref{P2:s_model}) are flexible semi-parametric models that include all commonly used generalized linear models as special cases.
These models imply that $Y \ind \bX \given \bbeta_0'\bX$ and $S \ind \bX \given \balpha_0'\bX$, in general, so that $\bbeta_0'\bX$ and $\balpha_0'\bX$ fully capture the dependencies of $Y$ and $S$ on $\bX$ respectively.
More importantly, the models (\ref{P2:y_model}) and (\ref{P2:s_model}) imply that $Y \ind S \given \bX$, and yet $S$ is dependent on $\bX$. Thus, in some sense, $S$ behaves as a so-called `instrumental variable' \citep{Bowden_1990, Pearl_2000} under our setting.
\begin{remark}\label{P2:assmpn_remark}
\emph{
Note that the condition $Y \ind S \given \bX$ 
does \emph{not} contradict in any way the extreme surrogacy assumption (\ref{P2:misclass_error}) which only relates $Y$ and $S$ \emph{marginally}, and only so in the \emph{tails} of $S$.
Moreover, this condition 
is very different from the typical surrogacy assumption adopted in the literature on measurement error and misclassification \citep{Carroll_2006, Buonaccorsi_2010}, namely 
$S \ind \bX \given Y$. If this holds in our case, $\balpha_0 \propto \bbeta_0$ must hold 
and the problem thus becomes trivial. Moreover, the typical measurement error assumption is often not realistic in the EMR setting. For example, when $S$ is the ICD9 code and $\bX$ consists of features such as medications for treating the disease, patients with a higher value of $S$ are likely to have higher values of $\bX$ among those with $Y=1$. Our assumption of $Y \ind S \given \bX$ is more suitable for the EMR setting, since $\bX$ is often comprehensive and the ICD9 code does not contribute any additional information on $Y$ 
beyond $\bX$. This assumption
ensures that the restriction $S \in \Isc_q$ underlying the construction of $\Usc_{\nq}$ and $\Fscr_{\nq}$ does not alter the relation between $Y$ and $\bX$ in (\ref{P2:y_model}) that defines our parameter of interest $\bbeta_0$.
}
\end{remark}

\subsection{Basic Foundations of Our Approach}\label{P2:basic_found}

We next discuss some useful motivations and essential fundamentals underlying our approach for recovering $\bbeta_0$.
For any given $q \in (0,1]$, let $p^*_q \equiv \E_q(Y^*_q) = 1/2$ 
and define:
\begin{eqnarray}
&& \L_q(\bv) \; = \; \E_q[\{Y - p_q - \bv'(\bX-\bmu_q)\}^2], \;\;\;\; \bbetabar_q \; = \; \underset{\bv \in \mathbb{R}^p}{\mbox{arg min}} \;\L_q(\bv); \quad \mbox{and} \label{P2:y_sqloss}\\
&& \L_q^*(\bv) \; = \; \E_q[\{Y^*_q - p_q^* - \bv'(\bX-\bmu_q)\}^2], \;\; \balphabar_q \; = \; \underset{\bv \in \mathbb{R}^p}{\mbox{arg min}}\; \L_q^*(\bv). \label{P2:s_sqloss}
\end{eqnarray}
Since only the coefficients corresponding to $\bX$ are of interest, we center all  variables in the definitions of $\L_q(\cdot)$ and $\L_q^*(\cdot)$.  With $\bSigma_q \succ 0$, both $\bbetabar_q$ and $\balphabar_q$ are clearly well-defined and unique.
While we focus on the squared loss throughout for convenience in constructing  the $\ULASSO$ estimator and ease of theoretical derivations, other convex loss functions more suited for binary outcomes, such as the logistic loss, can also be considered, but the corresponding technical analyses can be much more involved. We refer to Section \ref{P2:disc_oth_losses} 
for further discussions. Besides, it is also worth noting that least squares regression for binary outcomes is closely related to the well known linear discriminant analysis (LDA) approach. In fact, for binary outcomes, the least squares parameter vector is proportional to the LDA direction; see Chapters 4.2, 4.3 and Exercise 4.2 of \citet{Hastie_2008} for further details. 

The main motivation behind our consideration of (\ref{P2:y_sqloss}) and (\ref{P2:s_sqloss}) lies in an interesting result (Theorem 2.1) of \citet{Duan_Li_1989} which shows that for any outcome ${Y}$ satisfying a SIM given $\bX$ with some parameter $\bgamma \in \mathbb{R}^p$, if the following two conditions hold:
(i) $\E(\bv'\bX \given \bgamma'\bX)$ is a linear function of $\bgamma'\bX \;\; \forall \; \bv \in \R^p$, and
(ii) for a loss function $\Lsc({Y};a+\bv'\bX)$ that is convex in the second argument, if  $(\overline{a},\overline{\bv}')' = \argmin_{a,\bv} \E\{\Lsc({Y};a+\bv'\bX)\}$  exists and is unique,
then $\overline{\bv} \propto \bgamma$, i.e. $\overline{\bv}$ recovers 
$\bgamma$ up to a scalar multiple.
A similar result was also derived earlier by \citet{Brillinger_1982} for the special case when $\bX$ is Gaussian (so that condition (i) is automatically satisfied) and $\Lsc(u,v) = (u - v)^2$ corresponds to the squared loss.
In recent years, several works on signal recovery for SIMs in high dimensional settings have exploited and/or independently rediscovered this remarkable result of \citet{Duan_Li_1989}, including \citet{Thramp_2015, Neykov_L1_2016, Plan_Vershynin_2013_a, Plan_2016, Plan_2017}, among others, for the special case of Gaussian designs and the squared loss, as well as \citet{Goldstein_2016, Genzel_2017} and \citet{Wei_2018} for more general designs and/or loss functions.

Under our setting, as noted earlier, since $Y \ind S \given \bX$, the SIM (\ref{P2:y_model}) continues to hold even under the restriction $S \in \Isc_q$ that underlines the construction of the data subsets $\Usc_{\nq}$ and $\Fscr_{\nq}$. This, together with the results of \citet{Duan_Li_1989}, suggest that if $\Fscr_{\nq}$ were actually observed, a minimization of the corresponding empirical squared loss for $Y$ based on $\Fscr_{\nq}$ could potentially lead to a consistent estimator of the $\bbeta_0$ direction.

Of course, the major issue is that we only observe $\Usc_{\nq}$, and therefore can only hope to minimize the empirical squared loss for $Y^*_q$ based on $\Usc_{\nq}$, the empirical counterpart of $ \L_q^*(\bv)$ in (\ref{P2:s_sqloss}). However, owing to the extreme surrogacy assumption, $\L_q(\bv) - \L_q^*(\bv) \approx 0$ and therefore, due to the smoothness and convexity of $\L_q(\cdot)$ and $\L_q^*(\cdot)$, their minimizers $\balphabar_q$ and $\bbetabar_q$ are expected to be close.
Lastly, another critical issue is the validity of the condition (i) above on `linear conditional expectations' for the underlying design distribution which, in our case, is that of $\bX \given S \in \Isc_q$ and \emph{not} that of $\bX$ itself. Even if it holds for the distribution of $\bX$, it is unlikely to hold for that of $\bX \given S \in \Isc_q$, especially for small enough $q$. We therefore assume a different kind of a condition, involving \emph{only the marginal distribution} of $\bX$, that is more reasonable and likely to hold in practice for a fairly wide class of distributions.
%
\begin{assumption}[Design Linearity Condition - Linear Conditional Expectation]\label{P2:dlc_assmpn}
\emph{
We assume that for any $\bv \in \mathbb{R}^p$, $\E(\bv'\bX \mid \balpha_0'\bX, \bbeta_0'\bX)$ is linear in $\balpha_0'\bX$ and $\bbeta_0'\bX$, that is
\begin{alignat}{2}
& \E(\bv'\bX\given\balpha_0'\bX,\bbeta_0'\bX) \; && = \; c_{\bv} + a_{\bv}(\balpha_0'\bX) + b_{\bv}(\bbeta_0'\bX), \quad \mbox{and} \label{P2:lin_condn_eqn_1}\\
& \E(\bbeta_0'\bX \given \balpha_0'\bX) \; && = \; \overline{c} + \overline{a}(\balpha_0'\bX), \label{P2:lin_condn_eqn_2}
\end{alignat}
for some constants  $(c_{\bv}, a_{\bv}, b_{\bv})$ depending on $\bv$, and some constants $(\overline{c}, \overline{a})$.}
\end{assumption}
\begin{remark}\label{P2:dlc_remark}
\emph{
Note that the conditions (\ref{P2:lin_condn_eqn_1}) and (\ref{P2:lin_condn_eqn_2}) are restrictions on the original design distribution, the unconditional distribution of $\bX$, and do \emph{not} involve $Y$, $Y_q^*$ or $1(S \in \Isc_q)$. Conditions of a similar flavor are commonly adopted in the sufficient dimension reduction literature, including SIMs as special cases; see \citet{K-C_Li_1991}, \citet{Cook_2009} and references therein for further discussions on such conditions and their applicability. Our condition (\ref{P2:lin_condn_eqn_1}) is slightly stronger than the typical design linearity conditions assumed in the SIM literature \citep{Duan_Li_1989,Duan_1991, Wang_2012} in the sense that it requires joint linearity in $\balpha_0'\bX$ and $\bbeta_0'\bX$, instead of only $\bbeta_0'\bX$. This is imposed because the distribution of $\{\bX \given S \in \Isc_q\}$ inherently depends on $\balpha_0'\bX$ through $S$ owing to (\ref{P2:s_model}). Nevertheless, both (\ref{P2:lin_condn_eqn_1}) and (\ref{P2:lin_condn_eqn_2}) are satisfied, for instance, by all elliptically symmetric distributions \citep{Goldstein_2016, Wei_2018} including the Gaussian distribution, among others.
Moreover, \citet{Hall_1993} have also argued that for a wide class of distributions satisfying mild restrictions, such design linearity conditions are `approximately true' with high probability for most directions $\bv \in \mathbb{R}^p$, as long as $p$ is large enough; see also the results of \citet{Diaconis_1984}. Lastly, a more careful inspection of the proof of Theorem \ref{P2:THM_1} below, which is where Assumption \ref{P2:dlc_assmpn} is actually required, will reveal that the condition (\ref{P2:lin_condn_eqn_1}) is needed to hold for \emph{only one} specific choice of $\bv$, given by $\bv = \bbetabar_q$ as in (\ref{P2:y_sqloss}), and not for all $\bv \in \R^p$. Nevertheless, we stick to the stronger, but more familiar, formulation in (\ref{P2:lin_condn_eqn_1}).
}
\end{remark}

The explicit relationships between $\{\bbetabar_q,\balphabar_q\}$ in (\ref{P2:y_sqloss})-(\ref{P2:s_sqloss}) and the original SIM parameters  $\{\bbeta_0,\balpha_0\}$ in (\ref{P2:y_model})-(\ref{P2:s_model}) are given by the following result.
\begin{theorem}\label{P2:THM_1}
Assuming the design linearity conditions (\ref{P2:lin_condn_eqn_1})-(\ref{P2:lin_condn_eqn_2}),
\begin{eqnarray}
\bbetabar_q &=& a_q\balpha_0 + b_q\bbeta_0, \quad \mbox{and}\label{P2:thm1_eqn1}\\
\balphabar_q &=& a_q^*\balpha_0 \label{P2:thm1_eqn2},
\end{eqnarray}
where $a_q = a_{\bbetabar_q}$, $b_q=b_{\bbetabar_q}$,  and $a_q^*= a_{\balphabar_q} + b_{\balphabar_q}\overline{a}$. 
\end{theorem}
Thus, a simple minimization of the empirical squared loss for $Y^*_q$ based on $\Usc_{\nq}$ would only recover the direction of $\balpha_0$, and \emph{not} that of $\bbeta_0$. This makes sense since $(Y^*_q \mid \bX, S \in \Isc_q)$ follows a SIM with parameter $\balpha_0$ after all. On the other hand, (\ref{P2:thm1_eqn1}) shows that even if $Y$ were observed, the estimator obtained from a simple minimization of the empirical squared loss for $Y$ based on $\Fscr_{\nq}$ would only recover the direction of $\bbetabar_q$, and not the $\bbeta_0$ direction itself. This is largely due to the restriction of $S \in \Isc_q$, which leads to the distribution of $\bX \given S \in \Isc_q$, the underlying design distribution, to depend on $\balpha_0$ and makes the conventional linearity condition, requiring $\E_q(\bv'\bx \given \bbeta_0'\bX)$ to be linear in $\bbeta_0'\bX$, unlikely to hold.

The proof of Theorem \ref{P2:THM_1} can be found in Appendix \ref{P2:pf_thm1}.
Theorem \ref{P2:THM_1} and the subsequent discussions above actually continue to hold \emph{even if} the squared loss in (\ref{P2:y_sqloss})-(\ref{P2:s_sqloss}) is replaced by any other convex loss provided the corresponding minimizers $\bbetabar_q$ and $\balphabar_q$ exist and are unique, although we focus only on the squared loss here which suffices for all our purposes.
%
Theorem \ref{P2:THM_1}, in its generality, therefore shows that some further assumptions on the \emph{structure} of $\bbeta_0$ and $\balpha_0$ are clearly needed in order to estimate $\bbeta_0$ based on $\Usc_{\nq}$ (or even $\Fscr_{\nq}$).

\def\Yhat{\widetilde{Y}}
\def\bXhat{\widetilde{\bX}}

\section{The Unsupervised LASSO Estimator}\label{P2:ULASSO}

We next demonstrate that if $\bbeta_0$ additionally satisfies some (structured) sparsity assumptions, then it is possible to recover $\bbeta_0$ from $\Usc_{\nq}$ based on our proposed $\ULASSO$ estimator obtained via an $L_1$-penalized least squares regression of $Y^*_q$ on $\bX$.
Note that our motivations for enforcing sparsity, mainly in the light of Theorem \ref{P2:THM_1}, are however quite \emph{different} from those typical in the high dimensional statistics literature, although even under our setting, it may still help if $\bX$ is high dimensional (compared to the effective sample size $\nq$). Lastly, it is also worth noting that sparsity is a scale invariant criteria and hence, fits well into our SIM based framework in (\ref{P2:y_model})-(\ref{P2:s_model}) with $\bbeta_0$ identifiable only upto scalar multiples.

\subsection{The Estimator}
Let $\bXbar_{\nq}$ and $\Ybar^*_{\nq}$ denote the sample means of $\bX$ and $Y_{q}^*$ in $\Usc_{\nq}$ respectively, and let $\Yhat_{q,i}^* = Y_{q,i}^* - \Ybar^*_{\nq}$, $\bXhat_{q,i} = \bX_i - \bXbar_{\nq}$, and $\widehat{\bSigma}_q = {\nq}^{-1} \sum_{i=1}^{\nq} \bXhat_{q,i} \bXhat_{q,i}'$. For any $\bbeta, \bv \in \R^p$,  define: {
\begin{align}
& \Lsc_{\nq}(\Usc_{\nq};\bbeta) = \frac{1}{\nq} \sum_{i=1}^{\nq} (\Yhat_{q,i}^* - \bbeta'\bXhat_{q,i})^2, \;\;
 \bT_{\nq}(\bbeta) = \frac{1}{\nq} \sum_{i=1}^{\nq} \bXhat_{q,i}(\Yhat^*_{q,i}- \bbeta'\bXhat_{q,i}), \label{P2:esql}
\end{align}}
and write $\T_{\nq} \equiv \bT_{\nq}(\bbetabar_q)$. The centering in (\ref{P2:esql}) allows us to remove the nuisance intercept parameter.  Assuming that $\bbeta_0$ is indeed sparse, we then propose to estimate the $\bbeta_0$ direction based on the Unsupervised LASSO
($\ULASSO$)  estimator, defined as follows:
\begin{equation}
\bbetahat_{\nq}(\lambda)  \; \equiv \;  \bbetahat_{\nq}(\lambda;\Usc_{\nq}) \; = \; \underset{\bbeta\in \mathbb{R}^p}{\mbox{arg\;min}} \left\{\Lsc_{\nq}(\Usc_{\nq};\bbeta) + \lambda \|\bbeta\|_1\right\}, \label{P2:ULASSO_est}
\end{equation} where $\lambda \geq 0$ denotes the tuning parameter controlling the extent of the $L_1$ penalization. Below we study finite sample properties of $\bbetahat_{\nq}(\lambda)$ in terms of deterministic deviation bounds, followed by probabilistic bounds regarding performance guarantees and convergence rates.


\subsection{Theoretical Properties}
We first provide finite sample deterministic deviation bounds for the proposed $\ULASSO$ estimator in Theorem \ref{P2:THM_2} below, under Assumptions \ref{P2:strngconv_assmpn} and \ref{P2:sparsity_assmpn} given in Appendix \ref{P2:assumption}. 
The proof of Theorem \ref{P2:THM_2} can be found in Appendix \ref{P2:pf_thm2}.

\begin{theorem}\label{P2:THM_2}
Under Assumption \ref{P2:strngconv_assmpn} and condition C1 of Assumption \ref{P2:sparsity_assmpn}, for any realization of $\Usc_{\nq}$ and any choice of $\lambda$ in (\ref{P2:ULASSO_est}) such that $\lambda \geq 4\| \T_{\nq} \|_{\infty}$, with $\T_{\nq}$ as defined in (\ref{P2:esql}), and
$\lambda$ satisfies condition C2 of Assumption \ref{P2:sparsity_assmpn}, the $\ULASSO$ estimator $\bbetahat_{\nq}(\lambda)$ satisfies:
\begin{equation}
\left\| \bbetahat_{\nq}(\lambda) - b_q \bbeta_0 \right\|_2 \; \leq \;  \frac{\lambda}{\kappa_q} \left[ \left\{ 9 s_{\bbeta_0} + d_1 (\balpha_0, \bbeta_0) \right\}^{\half} + d_2(\balpha_0, \bbeta_0) \right], \quad \label{P2:thm2_dev_bound}
\end{equation}
where $b_q$ is as in (\ref{P2:thm1_eqn1}), $s_{\bbeta_0} = \| \bbeta_0 \|_0$ , $\kappa_q$ is a `restricted strong convexity' constant defined explicitly in Assumption \ref{P2:strngconv_assmpn}, and $d_1(\balpha_0, \bbeta_0)$, $d_2(\balpha_0, \bbeta_0) > 0$ are constants 
given by: 
\begin{align*}
& d_1(\balpha_0,\bbeta_0) \; =   \; 4 \dbar(\balpha_0, \bbeta_0) \left\| \Pi^c_{\bbeta_0} (\balpha_0) \right\|_1 \; \mbox{and} \;\; d_2(\balpha_0,\bbeta_0)  \; =  \; \dbar(\balpha_0, \bbeta_0) \left\| \balpha_0 \right\|_2, \\ 
& \nonumber \mbox{where} \;\; \dbar(\balpha_0,\bbeta_0)  \; =  \; 4 \left\| \Pibac \right\|_1 + 3 s^{\half}_{\bbeta_0} \Cmax(\balpha_0, \bbeta_0)\Cmin(\balpha_0,\bbeta_0)^{-2},
\end{align*}
with $\Cmin(\balpha_0,\bbeta_0) = \min \{ |\balpha_{0 [j]}| : j \in \Asc^c(\bbeta_0) \cap \Asc(\balpha_0)\} > 0$ and $\Cmax(\balpha_0,\bbeta_0)$ $= \max \{ |\balpha_{0 [j]}|$ $: j  \in \Asc^c(\bbeta_0) \cap \Asc_0 (\balpha_0) \} > 0$. 
\end{theorem}
\begin{remark}\label{P2:thm1_remark}\emph{
Assumption \ref{P2:strngconv_assmpn} is a standard restricted eigenvalue assumption and Assumption \ref{P2:sparsity_assmpn} imposes some mild restrictions on the sparsity patterns of $\bbeta_0$, $\balpha_0$ and $\bbetahat_{\nq}(\lambda, \Usc_{\nq})$.
Apart from the universal constants and the strong convexity constant $\kappa_q$, the bound primarily depends on $\lambda$ whose order would determine the convergence rate of $\bbetahat_{\nq}(\lambda)$. The random lower bound $4 \| \T_{\nq} \|_{\infty}$ characterizing the choice of $\lambda$ in Theorem \ref{P2:THM_2} therefore becomes the quantity of primary interest. If we can find a non-random sequence $a_{\nq} \to 0$ at a satisfactorily fast enough rate, and $a_{\nq}$ can be shown to upper bound $\| \T_{\nq} \|_{\infty}$ with high probability, then a choice of $\lambda = 4a_{\nq}$, as long as it satisfies the additional conditions required for Theorem \ref{P2:THM_2}, will guarantee the bound in (\ref{P2:thm2_dev_bound}) to hold with high probability at the rate of $O(a_{\nq}/\kappa_q)$.
}
\end{remark}

To characterize the probabilistic performance guarantees, we next study the behavior of the lower bound $4\|\T_{\nq}\|_{\infty}$ of $\lambda$, as assumed in Theorem \ref{P2:THM_2}, under the assumption that the distribution of $\bX \given S \in \Isc_q$ follows a subgaussian distribution so that it has sufficiently well behaved tails, as detailed in Assumption \ref{P2:subgaussian_defn} given in Appendix \ref{P2:assumption}. 
 Theorem \ref{P2:COR_1} below provides a (sharp) probabilistic bound for $\|\T_{\nq}\|_{\infty}$ establishing its convergence rates. A more general and detailed version of this result is available in Theorem \ref{P2:THM_3} (stated in Appendix \ref{P2:thm3.2_proof}), 
 where we obtain explicit finite sample tail bounds for $\|\T_{\nq}\|_{\infty}$. The proofs of Theorems \ref{P2:COR_1} and \ref{P2:THM_3} can be both found in Appendix \ref{P2:thm3.2_proof}.


\begin{theorem}\label{P2:COR_1} Let $\bc=(c_1, \hdots, c_6)' > 0$ be any set of universal positive
constants such that $\mbox{max}\;(c_1, c_2) > 1$ and $c_4, c_5 > 1$. Let $c_0 = (c_4 + c_5 c_6)$ and assume WLOG that $\pi_q < 1/2$. 
Then, under Assumption \ref{P2:subgaussian_defn}, we have: with probability at least
$$1 - \left(\frac{\pi_q}{1-\pi_q}\right)^{c_3} - \frac{2}{p^{(c_1-1)}\nq^{(c_2-1)}} - \frac{2}{p^{(c_4-1)}} - \frac{2}{p^{(c_5-1)}} - \frac{2}{p^{c_6}},$$
\begin{equation}\label{P2:thm3_probbound_2}
\left\|\T_{\nq}\right\|_{\infty} \; \leq \; a_{\nq} \; \equiv \; a_{\nq}(\bc),
\end{equation}
where, with the constants $(\sigma_q, \gamma_q) > 0$ defined explicitly in Theorem \ref{P2:THM_3}, $a_{\nq}(\bc)$ is given by:
$$ \sigma_q \sqrt{2 \log (p^{c_1} \nq^{c_2 })} \left\{\pi_q + \sqrt{\frac{(1 -2 \pi_q)c_3}{\nq}}\right\} + 2\sigma_q\gamma_q \left( \sqrt{8 c_4\frac{\log p}{\nq}} + \frac{\log p}{\nq} c_0 \right).$$
\end{theorem}
Theorem \ref{P2:COR_1} implies that for some suitably chosen constants $\bc$, setting $\lambda = 4 a_{\nq}$ ensures that the condition $\lambda \geq 4 \| \T_{\nq} \|_{\infty}$, required for Theorem \ref{P2:THM_2}, holds with high probability. Consequently, with  $\lambda = 4 a_{\nq}$, as long as it satisfies the other conditions required for Theorem \ref{P2:THM_2}, the deviation bound (\ref{P2:thm2_dev_bound}) holds with high probability as well, thereby ensuring a  convergence rate of $O(a_{\nq}/\kappa_q)$ for $\bbetahat_{\nq}(\lambda)$ as an estimator of the $\bbeta_0$ direction.

\begin{remark}[Convergence Rates and Other Implications]\label{P2:thm2_remark}
\emph{
Theorem \ref{P2:COR_1} applies generally to \emph{any} $\pi_q$, not necessarily behaving as in (\ref{P2:misclass_error}). Of course, the bound is most useful if $\pi_q$ is polynomial in $q$. 
Turning to the convergence rate of $a_{\nq}$ itself, we note that the (dominating) polynomial part of the rate is determined primarily by $\pi_q$ and $\nq^{-1/2}$, which behave antagonistically with respect to each other as $q$ varies, so that the rate exhibits an interesting phenomenon similar to a \emph{variance-bias tradeoff}. The misclassification error $\pi_q$, expected to increase as $q$ increases, can be viewed as a `bias' term, while $\nq^{-1/2}$, which decreases as $q$ increases, corresponds to the usual variance (rather, standard deviation) term. In particular, with $\pi_q = O(q^{\nu})$ for some given $\nu > 0$, as in (\ref{P2:misclass_error}), and $q = O(N^{-\eta})$ for some unknown $\eta \in (0,1)$, the combined rate: $(\pi_q + \nq^{-1/2})$ $\equiv O\{N^{-\nu\eta} + N^{-(1-\eta)/2}\}$ can be minimized with respect to 
$\eta$, leading to an optimal choice given by: $\eta_{opt} = 1/(2\nu +1)$, and a corresponding optimal order of $q$ given by: $q_{opt} =  O\{N^{-1/(2\nu +1)}\}$. For $q = q_{opt}$, $\pi_q$ and $\nq^{-1/2}$ have the same order, so that the optimal order of the (polynomial part of the) convergence rate of $a_{\nq}$ is given by: $(a_{\nq})_{opt} = O\{N^{-\nu/(2\nu + 1)}\}$.
}
\end{remark}

\subsection{Practical Choice of the Tuning Parameter}
The theoretical choice of $\lambda=4 a_{\nq}$, which is of order $O[\{\log (\nq p)\}^{1/2}$ $(\pi_q + \nq^{-1/2})]$, is not quite feasible for implementing $\bbetahat_{\nq}(\lambda)$ in practice, since $a_{\nq}$ involves $\pi_q$, as well as the constants $\sigma_q$ and $\gamma_q$, which are all typically unknown. To this end, we note that owing to the additional $\pi_q$ term, as well as the $\log(\nq p)$ term, in $a_{\nq}$, the order of the appropriate choice of $\lambda$ under our setting is expected to be slightly \emph{higher} than the standard order of $O [\{ (\log  p)/\nq\}^{1/2} ]$ for typical $L_1$-penalized estimation. This is because sparser solutions are favored in the current setting. Motivated by this intuition, we propose to choose $\lambda$ in practice through minimizing a criteria similar to the Bayes Information Criteria (BIC), $ \mbox{BIC}(\lambda)$, defined as follows:
 \begin{equation}
\mbox{BIC}(\lambda) \; \equiv \; \mbox{BIC}\{\lambda; \bbetahat_{\nq}(\lambda); \Usc_{\nq}\} \; = \; \Lsc_{\nq}\{\Usc_{\nq}; \bbetahat_{\nq}(\lambda)\} + \frac{\log(\nq)}{\nq} \| \bbetahat_{\nq}(\lambda)\|_0 . \label{P2:bic_criteria}
\end{equation}
Compared to other standard tuning parameter selection criteria, such as Akaike Information Criteria (AIC) and cross-validation (CV), BIC tends to select sparser solutions which serves well for our purpose. While a detailed theoretical analysis is beyond the scope of this paper, we find that, based on our numerical studies, the above criteria works quite well in practice.


\section{Numerical Results: Simulation Studies}\label{P2:sim}

We conducted extensive simulation studies to examine the performance of the U$_{\mbox{\tiny LASSO}}$ estimator with $N = 100,000$, and compare it to that of a supervised logistic LASSO (S$_{\mbox{\tiny LASSO}}$) estimator obtained by fitting $L_1$-penalized logistic regressions \citep[Chapt. 4.4]{Hastie_2008} to
$n =$ $300$ and $500$ sized labeled data. While other supervised estimators were also considered, in the interest of space, we only report the results for S$_{\mbox{\tiny LASSO}}$ which is expected to have the most competitive performance as it exploits sparsity as well as knowledge of the true link function.  We consider $p = 20, 50$, and $q= 0.02, 0.04$. We generate $\bX \sim \Nsc_p (\bzero, \bSigma)$, where $\bSigma \equiv \bSigma_{\rho} = (\rho^{|i-j|})_{i=1,...p}^{j=1,...,p}$ with $\rho = 0$ or $0.2$. Given $\bX$, we generate $S$ and $Y$ as:
\begin{align*}
& S \; = \; \balpha_0'\bX + \epsilon^* \quad \mbox{with} \;\; \epsilon^* \sim \Nsc_1(0,1) \;\; \mbox{and} \;\; \epsilon \ind \bX, \quad \mbox{and} \\
& Y \; = \; 1(\bbeta_0'\bX + \epsilon > 0)  \quad \mbox{with} \;\; \epsilon \sim \mbox{Logistic}(0,1) \;\; \mbox{and} \;\; \epsilon \ind (S,\bX),
\end{align*}
where $\bbeta_0 = (\bone'_{c_p}, 0.5*\bone'_{c_p}, \bzero'_{p-2c_p})'$, $c_p = \lfloor p^{\half} \rfloor$, $\balpha_0 = \bbeta_0 + \bxi/(\log  N)$ with $\{ \bxi_{[j]}\}_{j=1}^p$ being $p$ \emph{fixed} realizations from either (I) $\Nsc_1(3, 1)$, or (II) $\mbox{Uniform}(2, 5)$. Such choices of $\balpha_0$ ensure that $\balpha_0$ is `close' to $\bbeta_0$, and yet its deviations from $\bbeta_0$ are of order $O\{1/\log(N)\}$, so that they are not `too close'. See
Appendix \ref{P2:splcase} (and Remark \ref{P2:splcase_rem2} in particular) 
for further insights regarding the rationale behind such choices, as well as a detailed theoretical analysis, with discussions, for a commonly adopted class of models for $(Y,S,\bX')'$, including those used here. 
The results are summarized based on 500 simulated datasets for each configuration. 

For any estimator $\bbetatil$, we consider its normalized (in both length and sign) version $\bbetatil^{\dag}$ such that $\|\bbetatil^{\dag}\|_2 = 1$ and $\bbeta_0'\bSigma \bbetatil^{\dag} \geq 0$. We also use the $\balpha_0$ direction, $\balpha_0^{\dag} = \balpha_0/\|\balpha_0\|_2$, as a benchmark estimator\footnote{
$\balpha_0^{\dag}$ essentially represents a `baseline' estimator of $\bbeta_0^{\dag}$ and corresponds to the case where one chooses to use \emph{all} observations of $S$ from the full original data $\Usc_N^*$ (ignoring the fact that $S$ can lead to an effective surrogate of $Y$ only when it assumes extreme values) for estimating $\bbeta_0^{\dag}$ and for classifying $Y$.}
of the $\bbeta_0$ direction, $\bbeta_0^{\dag}=\bbeta_0/\|\bbeta_0\|_2$.  For the U$_{\mbox{\tiny LASSO}}$ estimator, the tuning parameter was selected using the BIC criteria defined in (\ref{P2:bic_criteria}). For the supervised estimator, the tuning parameter was selected using the appropriate BIC criteria based on the logistic loss. All penalized estimators were implemented using the \texttt{R} package  \texttt{`glmnet'}.

\begin{table}[!ht]
\centering
\centerline{(a) Relative Efficiency of $\ULASSO$}
\begin{tabular}{crc| rr|r| rr|r}\hline\hline
& & &\multicolumn{3}{c|}{$p=20$}    &\multicolumn{3}{c}{$p=50$}\\ \hline
\multicolumn{3}{c|}{Settings} &\multicolumn{2}{c|}{S$_{\mbox{\tiny LASSO}}$ } &\raisebox{-2ex}[0cm][0cm]{$\balpha_0^{\dag}$}   &\multicolumn{2}{c|}{S$_{\mbox{\tiny LASSO}}$ } &\raisebox{-2ex}[0cm][0cm]{$\balpha_0^{\dag}$} \\
& $\rho$ & $q$ &{$300$}&{$500$}   &&{$ 300$}&{$500$}  &{} \\ \hline
  & $0$ & $.02$    	& 4.16 & {2.22} & 5.60     & 6.89 & {3.16} & 6.88 \\
\raisebox{-2ex}[0cm][0cm]{(I)} & $0$ & $.04$   	& 8.05 & {4.03} & 11.25    & 6.07 & {2.89} & 6.96 \\
& $.2$ & $.02$ 			 & 3.36 & {1.78} & 4.41     & 11.63 & {5.99} & 17.00 \\
& $.2$ & $.04$ & 2.49 & {1.46} & 3.67     & 5.14  & {2.57} & 6.38  \\ \hline
& $0$ & $.02$   & 6.97 & {3.48} & 10.46    & 6.10 & {2.84}  & 8.58 \\
\raisebox{-2ex}[0cm][0cm]{(II)} & $0$ & $.04$   & 7.23 & {3.91} & 8.54     & 5.95 & {2.67}  & 8.39 \\
& $.2$ & $.02$ & 4.44 & {2.48} & 7.67     & 5.14 & {2.57}  & 7.03 \\
& $.2$ & $.04$ & 4.20 & {2.25} & 6.15     & 7.05 & {3.42}  & 10.79\\ \hline
\end{tabular}\vspace{.1in}

\centerline{(b) AUC}
\begin{tabular}{crc| c|cc|c|c| c|cc|c|c }\hline\hline
& & &\multicolumn{5}{c|}{$p=20$}    &\multicolumn{5}{c}{$p=50$}\\ \hline
\multicolumn{3}{c|}{Settings}  & \raisebox{-2ex}[0cm][0cm]{U$_{\mbox{\tiny LASSO}}$} &\multicolumn{2}{c|}{S$_{\mbox{\tiny LASSO}}$ } &\raisebox{-2ex}[0cm][0cm]{$\balpha_0$} &\raisebox{-2ex}[0cm][0cm]{$\bbeta_0$}
    & \raisebox{-2ex}[0cm][0cm]{U$_{\mbox{\tiny LASSO}}$} &\multicolumn{2}{c|}{S$_{\mbox{\tiny LASSO}}$ } &\raisebox{-2ex}[0cm][0cm]{$\balpha_0$} &\raisebox{-2ex}[0cm][0cm]{$\bbeta_0$}\\
&$\rho$ & $q$ &&{$300$}&{$500$}  &{} &{}    &{} &{$300$}&{$500$}  &{} &{} \\ \hline\hline
&  $0$ & $.02$   & {.88} & .86 & {.87} & .86 & {.88}     & {.91} & .88 & {.90} & .88 &{.92} \\
\raisebox{-2ex}[0cm][0cm]{(I)} & $0$ & $ .04$   & {.88} & .86 & {.87} & .86 & {.88}    & {.91} & .88 & {.90} & .88 & {.92} \\
& $.2$ & $.02$ & {.90} & .89 & {.90} & .88 & {.90}     & {.94} & .91 & {.93} & .89 & {.94} \\
& $.2$ & $.04$ & {.90} & .89 & {.90} & .88 & {.90}     & {.93} & .91 & {.93} & .90 & {.94}  \\ \hline
& $0$ & $.02$   & {.88} & .86 & {.87} & .85 & {.88}    & {.91} & .88 & {.91}  & .87 & {.92} \\
\raisebox{-2ex}[0cm][0cm]{(II)} & $0$ & $.04$   & {.88} & .86 & {.87} & .86 & {.88}    & {.91} & .88 & {.90}  & .87 & {.92} \\
& $.2$ & $.02$ & {.90} & .89 & {.90} & .87 & {.90}     & {.93} & .91 & {.93}  & .89 & {.94} \\
& $.2$ & $.04$ & {.90} & .89 & {.90} & .88 & {.90}     & {.93} & .91 & {.93}  & .88 & {.94} \\ \hline
\end{tabular}
\caption{ (a) Efficiency, with respect to empirical MSE, of the U$_{\mbox{\tiny LASSO}}$ relative to the S$_{\mbox{\tiny LASSO}}$ estimator, obtained by fitting a logistic LASSO to $n=300$ and $500$ sized labeled data, as well as $\balpha_0^{\dag} = \balpha_0/\|\balpha_0\|_2$, under 
various settings; (b) out-of-sample AUC achieved by the U$_{\mbox{\tiny LASSO}}$, the S$_{\mbox{\tiny LASSO}}$ with $n = 300$ and $500$ labels, as well as $\balpha_0$ and $\bbeta_0$, under various settings.}
\label{tab-REauc}
\end{table}%
We report in Table \ref{tab-REauc}(a) the relative efficiency (RE) of the U$_{\mbox{\tiny LASSO}}$ estimator compared to other estimators in approximating $\bbeta_0^{\dag}$, with respect to the empirical mean squared error (MSE).
For any estimator $\bbetatil^{\dag}$ of $\bbeta_0^{\dag}$, the empirical MSE is given by $\| \bbetatil^{\dag} -\bbeta_0^{\dag}\|_2^2$ averaged over the 500 replications, and the RE between two such estimators is given by the inverse ratio of their respective empirical MSEs.
In Table \ref{tab-REauc}(b), we report the out-of-sample
classification performance of $\bbetatil'\bX$ based on the area under the receiver operating characteristic curve (AUC), a scale-invariant measure, for different choices of $\bbetatil$. (All out-of-sample measures were computed based on independent validation datasets of size $N = 100,000$). As a performance benchmark, we also report the oracle AUC
associated with $\bbeta_0$ (i.e. the AUC achieved by the oracle classifier $\bbeta_0'\bX$ that uses knowledge of the true $\bbeta_0$). Lastly, to compare the variable selection performance of the U$_{\mbox{\tiny LASSO}}$ and the $L_1$-penalized supervised estimators, we report in Table \ref{tab-TPRFPR} their corresponding average true positive rate (TPR) and false positive rate (FPR) with respect to $\Asc(\bbeta_0)$.

\begin{table}[!ht]
\centering
\centerline{(a) TPR in Variable Selection}
\begin{tabular}{crl|c|cc|c|cc}\hline\hline
&&&\multicolumn{3}{c|}{$p=20$}    &\multicolumn{3}{c}{$p=50$}\\ \hline
\multicolumn{3}{c|}{Settings} &\raisebox{-2ex}[0cm][0cm]{U$_{\mbox{\tiny LASSO}}$} &\multicolumn{2}{c|}{S$_{\mbox{\tiny LASSO}}$}
&\raisebox{-2ex}[0cm][0cm]{U$_{\mbox{\tiny LASSO}}$} &\multicolumn{2}{c}{S$_{\mbox{\tiny LASSO}}$}  \\
& $\rho$ & $q$ & &{$ 300$}&{$500$}     &{} &{$300$}&{$500$}  \\ \hline
& $0$ & $.02$   & {1.0} & .98 & {1.0}     & {1.0} & .89 & {.98} \\
\raisebox{-2ex}[0cm][0cm]{(I)}& $0$ & $.04$   & {1.0} & .97 & {1.0}     & {1.0} & .90 & {.98} \\
&$.2$ & $.02$ & {1.0} & .98 & {1.0}     & {1.0} & .93 & {.99} \\
&$.2$ & $.04$ & {1.0} & .98 & {1.0}     & {1.0} & .93 & {.99}  \\ \hline
&$0$ & $.02$   & {1.0} & .97 & {1.0}    & {1.0} & .90 & {.98} \\
\raisebox{-2ex}[0cm][0cm]{(II)}&$0$ & $.04$   & {1.0} & .98 & {1.0}    & {1.0} & .89 & {.98} \\
&$.2$ & $.02$ & {1.0} & .98 & {1.0}     & {1.0} & .93 & {.99} \\
&$.2$ & $.04$ & {1.0} & .97 & {1.0}     & {1.0} & .93 & {.99} \\ \hline
\end{tabular}\vspace{.1in}
\centerline{(b) FPR in Variable Selection}
\begin{tabular}{crl|c|cc|c|cc}\hline\hline
&&&\multicolumn{3}{c|}{$p=20$}    &\multicolumn{3}{c}{$p=50$}\\ \hline
\multicolumn{3}{c|}{Settings} &\raisebox{-2ex}[0cm][0cm]{U$_{\mbox{\tiny LASSO}}$} &\multicolumn{2}{c|}{S$_{\mbox{\tiny LASSO}}$}
&\raisebox{-2ex}[0cm][0cm]{U$_{\mbox{\tiny LASSO}}$} &\multicolumn{2}{c}{S$_{\mbox{\tiny LASSO}}$}  \\
& $\rho$ & $q$ & &{$ 300$}&{$500$}     &{} &{$300$}&{$500$}  \\ \hline
&$0$ & $.02$   & {.00} & .23 & {.25}     & {.00} & .12 & {.15} \\
\raisebox{-2ex}[0cm][0cm]{(I)}&$0$ & $.04$   & {.02} & .23 & {.23}     & {.01} & .12 & {.15} \\
&$.2$ & $.02$ & {.00} & .20 & {.20}     & {.00} & .10 & {.12} \\
&$.2$ & $.04$ & {.09} & .19 & {.19}     & {.02} & .11 & {.13}  \\ \hline
&$0$ & $.02$   & {.00} & .22 & {.24}    & {.00} & .12 & {.15} \\
\raisebox{-2ex}[0cm][0cm]{(II)}&$0$ & $.04$   & {.00} & .23 & {.23}    & {.00} & .11 & {.16} \\
&$.2$ & $.02$ & {.00} & .20 & {.21}     & {.00} & .11 & {.13} \\
&$.2$ & $.04$ & {.00} & .20 & {.20}     & {.00} & .11 & {.13} \\ \hline
\end{tabular}
\centering\caption{Average (a) TPR and (b) FPR in variable selection achieved by the U$_{\mbox{\tiny LASSO}}$, and the S$_{\mbox{\tiny LASSO}}$ estimator, obtained by fitting a logistic LASSO to $n=300$ and $500$ sized labeled data, under the various settings.}
\label{tab-TPRFPR}
\end{table}%

Overall, the U$_{\mbox{\tiny LASSO}}$ estimator performs well with respect to all the criteria we have considered, and is fairly robust to the choice of $q$ as well as the underlying correlation structure of $\bX$.
Interestingly, over all the settings considered, the U$_{\mbox{\tiny LASSO}}$ outperforms the supervised S$_{\mbox{\tiny LASSO}}$ estimators with $n=300$ or $500$, in estimating the $\bbeta_0$ direction, achieving lower MSE. Further, its prediction performance, as measured by the AUC, is satisfactorily close to the oracle AUC achieved by the true $\bbeta_0$ direction. The AUC achieved by the U$_{\mbox{\tiny LASSO}}$ is also generally higher or comparable to those achieved by the supervised estimators. The support recovery performance of the U$_{\mbox{\tiny LASSO}}$, in terms of both the TPR and the FPR, is also found to be near-perfect, especially for the TPR, over all the cases, and is again uniformly superior to those achieved by the corresponding supervised estimators. 

Lastly, the \emph{performance of $\balpha_0$}, with respect to all the criteria considered, is clearly seen to be significantly \emph{worse}, over all the cases, than those of the U$_{\mbox{\tiny LASSO}}$ as well as most of the supervised estimators, thereby indicating that while it is `close' to $\bbeta_0$, it is not close enough (or sparse enough) to be considered a reasonable estimator of the $\bbeta_0$ direction. This also indicates the benefits (and necessity) of considering our penalized estimation procedure focussing only on the extreme subsets of $S$ exploiting its surrogacy therein, and highlights the unsuitability of an approach where one chooses to estimate $\bbeta_0$ using all the observations of $S$ (in the full original data $\Usc_N^*$) ignoring that the surrogacy of $S$ holds only in its extreme tails. The latter approach essentially leads only to a (near-perfect) estimator of $\balpha_0$ which, however, may not be a reasonable estimator of $\bbeta_0$ at all, as shown in our simulation results.


\section{Real Data Analysis: Application to EMR Data} \label{P2:data_ex}

We applied our proposed method to an EMR study of rheumatoid arthritis (RA) conducted at the Partners HealthCare Systems \citep{Liao_2010, Liao_2013}. The full study cohort consists of $44, 014$ patients, and the goal is to develop an EMR phenotyping algorithm to classify a binary outcome of interest $Y$, defined as confirmed diagnosis of RA, based on $\bX$ consisting of 37 covariates, including both codified and narrative mentions of RA as well as several related disease conditions including Lupus, Juvenile RA, Psoriatic Arthritis, Polymyalgia Rheumatica etc., among others. The codified mentions are defined as the number of ICD9 codes and the narrative mentions count the number of times these clinical terms are mentioned in the physicians' narrative notes assessed via natural language processing (NLP). The covariates also include diagnostic testing results for various standard RA biomarkers such as anti-cyclic citrullinated peptide (anti-CCP), anti-tumor necrosis factor (anti-TNF), rheumatoid factor etc., as well as a variety of RA medications and some known RA related clinical conditions. Due to the high correlation between the codified and NLP mentions of RA, we collapsed the two variables into a single variable of RA mentions as their sums. All the count variables were further log-transformed as: $x \to \log(1+x)$, to increase stability of the model fitting. All covariates were further standardized to have unit variance with respect to the full data.

For a subset of $n=500$ patients, we have their outcome $Y$ ascertained via manual chart review by experts. This set is only used to train the supervised estimators and evaluate the performance of the $\ULASSO$ estimator. We choose the surrogate $S$ as the total count of ICD9 diagnostic codes for RA taken at least a week apart, denoted by RA$_{\mbox{\tiny ICD9,w}}$.  It is natural to expect that when RA$_{\mbox{\tiny ICD9,w}}$ assumes too high or too low values, the patient is very likely to have RA or no RA respectively.  We let $q = 0.02$ and $0.05$, resulting in  $(\delta_q, \bar{\delta}_q) =$ (0, 45) and (0, 70), and $n_q = 4375$ and $5040$, respectively and the corresponding $\pi_q$ is 
estimated as $0.025$ and $0.053$ respectively, based on the available labeled data.

In addition to the U$_{\mbox{\tiny LASSO}}$ estimators for $q = 0.02$ and $0.05$, we further constructed a `combined' U$_{\mbox{\tiny LASSO}}$ estimator, obtained by averaging the two normalized U$_{\mbox{\tiny LASSO}}$ estimators. We also obtained the supervised LASSO and adaptive LASSO (ALASSO) penalized logistic regression estimators based on the labeled data. The tuning parameters for all the U$_{\mbox{\tiny LASSO}}$ estimators were selected using the BIC criteria in (\ref{P2:bic_criteria}), while those for the supervised logistic LASSO and ALASSO estimators were selected using the appropriate BIC criteria based on the logistic loss. We also estimated $\balpha_0$ by regressing 
the surrogate RA$_{\mbox{\tiny ICD9,w}}$ on $\bX$ based on the entire available data, via Poisson 
%
\begin{table}[H]
\centering
\begin{tabular}{r | rrr| rr| r}\hline\hline
&\multicolumn{3}{c|}{U$_{\mbox{\tiny LASSO}}$} &\multicolumn{2}{c|}{Supervised Estimators} &\raisebox{-2.7ex}[0cm][0cm]{$\widehat{\balpha}$}\\
\cline{1-7}
\raisebox{2.7ex}[0cm][0cm]{Predictors} &{$.02$}&{$.05$}&{Ave.}   &{LASSO}&{ALASSO (SE)}    & \\
  Age                			& $0$     & $0$     & $0$     & $0$      & $0 (.12)$    & $.06$ \\
  Gender             		& $0$     & $0$     & $0$     & $0$      & $0 (.03)$    & $-.01$ \\
  RA Mentions & $.81$ & $.89$ & $.85$ & $.85$  	& $.85 (.14)$    & $.69$ \\
  {Psoriatic Arthritis$_{\mbox{\tiny ICD9}}$}       		& $0$     & $0$     & $0$     & $-.02$  & $0 (.05)$   & $-.01$ \\
  Psoriatic Arthritis$_{\mbox{\tiny ICD9,w}}$              & $0$     & $0$     & $0$     & $0$      & $0 (.04)$  & $.05$ \\
  {Juvenile RA$_{\mbox{\tiny ICD9}}$}       		& $0$     & $0$     & $0$     & $-.08$ & $0 (.06)$    & $.21$ \\
  Lupus$_{\mbox{\tiny ICD9}}$              		& $0$     & $0$     & $0$     & $0$      & $0 (.07)$   & $.16$ \\
  Lupus$_{\mbox{\tiny ICD9,w}}$              & $0$     & $0$     & $0$     & $0$      & $0 (.09)$   & $-.09$ \\
  Juvenile RA$_{\mbox{\tiny ICD9,w}}$              & $0$     & $0$     & $0$     & $0$      & $0 (.09)$  & $-.26$ \\
  Spondylo Arthritis$_{\mbox{\tiny ICD9}}$ 		& $0$     & $0$     & $0$     & $0$      & $0 (.03)$    & $.02$ \\
  Polymyositis$_{\mbox{\tiny ICD9}}$		& $0$     & $0$     & $0$     & $0$      & $0 (.01)$    & $.01$ \\
  Dermatomyositis$_{\mbox{\tiny ICD9}}$    & $0$     & $0$     & $0$     & $0$      & $0 (.01)$    & $-.01$ \\
  Polymyalgia Rheumatica $_{\mbox{\tiny ICD9}}$                		& $0$     & $0$     & $0$     & $0$      & $0 (.03)$    & $.01$ \\
  Psoriatic Arthritis$_{\mbox{\tiny NLP}}$            		& $0$     & $0$     & $0$     & $0$      & $0 (.09)$ & $-.04$ \\
  {Juvenile RA$_{\mbox{\tiny NLP}}$}            	& $0$     & $0$     & $0$     & $-.07$ & $-.12 (.10)$  & $.02$ \\
  Lupus$_{\mbox{\tiny NLP}}$            		& $0$     & $0$     & $0$     & $0$      & $0 (.08)$  & $-.09$ \\
  {Methotrexate}       	& $0$     & $.12$ & $.06$ & $.21$    & $.22 (.14)$    & $.12$ \\
  {Anti-TNF}           	& $0$     & .01 & $.01$ & $.14$  & $.08 (.10)$    & $.02$ \\
  {Enbrel}                		& $.13$ & $.07$ & $.10$ & $0$   & $0 (.04)$    & $.07$ \\
  Humira               		& $0$     & $0$     & $0$     & $0$      & $0 (.03)$    & $-.00$ \\
  {Infliximab}                	& $.04$ & $0$     & $.02$ & $0$    & $0 (.05)$    & $.01$ \\
  Abatacept             			& $0$     & $0$     & $0$     & $0$      & $0 (.03)$    & $.01$ \\
  Rituximab                			& $0$     & $0$     & $0$     & $0$      & $0 (.02)$    & $-.04$ \\
  Anakinra               			& $0$     & $0$     & $0$     & $0$      & $0 (.03)$    & $.01$ \\
  Sulfasalazine               			& $0$     & $0$     & $0$     & $0$      & $0 (.05)$    & $.06$ \\
  Azathrioprine                			& $0$     & $0$     & $0$     & $0$      & $0 (.04)$    & $-.03$ \\
  Hydroxycholorquine               			& $0$     & $0$     & $0$     & $0$      & $0 (.02)$    & $.06$ \\
  {Leflunomide}              		& $.25$ & $.10$ & $.18$ & $.02$  	& $0 (.06)$    & $.05$ \\
  Penicillamine                			& $0$     & $0$     & $0$     & $0$      & $0 (.06)$    & $-.00$ \\
  {Gold Salts}                		& $.10$ & $.03$ & $.07$ & $0$      & $0 (.05)$    & $-.01$ \\
  Cyclosporine                			& $0$     & $0$     & $0$     & $0$      & $0 (.03)$    & $-.03$ \\
  Other Medications        		& $0$     & $0$     & $0$     & $0$      & $0 (.06)$    & $.02$ \\
  {Anti-CCP}           	& $0$     & $0$     & $0$     & $.04$  & $0 (.09)$    & $.01$ \\
  {Rheumatoid Factor}                 		& $0$     & $0$     & $0$     & $.08$  & $.06 (.10)$    & $.09$ \\
  {Erosion}            	& $.47$ & $.40$ & $.44$  & $.28$  & $.29 (.14)$    & $.06$ \\
  {Seropositive}       	& $.18$ & $.05$ & $.12$ & $.33$  & $.35 (.10)$    & $.03$ \\
  {Facts}              		& $.08$ & $.15$ & $.11$ & $0$    & $0 (.02)$    & $.56$ \\
  \hline
  {AUC}           & {.95} & {.94} & {.95} & {.95} & {.95}  & {.92} \\ \hline
\end{tabular}
\centering
\caption{Coordinate-wise comparison of the U$_{\mbox{\tiny LASSO}}$ estimator at $q = 0.02, 0.05$ along with an  average of these two estimators (Ave.), to the supervised LASSO and adaptive LASSO (ALASSO) penalized logistic regression estimators, obtained from $n=500$ sized validation data, as well as to the estimated $\balpha_0$ direction, $\widehat{\balpha}$, for the data example. Shown also are the bootstrap based SE estimates for the ALASSO estimator for reference, as well as the AUC estimates for all the estimators.}
\label{P2:tab_data_ex} 
\end{table}%
\noindent regression. All estimators were further normalized to have unit $L_2$ norm. Since bootstrap is expected to work for the ALASSO but not the LASSO regularized estimators, we also computed bootstrap based standard error (SE) estimates for the ALASSO estimator. Further, in order to examine the classification performance of all the estimators, we also obtained estimates of their corresponding their 
AUC measures using the labeled data. For the supervised estimators, $5-$fold CV was used to estimate the AUC to correct for over-fitting bias.

The results, shown in Table \ref{P2:tab_data_ex}, demonstrate the utility of the proposed U$_{\mbox{\tiny LASSO}}$ estimator.
First, the support and magnitudes of the U$_{\mbox{\tiny LASSO}}$ estimators are not too sensitive to the choice of $q$. The U$_{\mbox{\tiny LASSO}}$ for $q = 0.05$, and hence the `combined' U$_{\mbox{\tiny LASSO}}$, does select one or two more clinically relevant variables, including anti-TNF and methotrexate, thereby indicating the potential utility, at least in this case, of considering multiple choices of $q$ for constructing U$_{\mbox{\tiny LASSO}}$, followed by combining the estimators appropriately.
Moreover, the U$_{\mbox{\tiny LASSO}}$ estimators are all reasonably close to the supervised estimators from the $500$ sized labeled data. The differences between the U$_{\mbox{\tiny LASSO}}$  and the supervised estimators are mostly small in magnitude as compared to the estimated SEs. The U$_{\mbox{\tiny LASSO}}$ and the supervised estimators also achieved nearly identical classification performance in terms of the AUC. Finally, the performance of the $\balpha_0$ estimator, both in terms of estimation, as well as prediction based on the AUC measure, is substantially \emph{worse} than those of all the U$_{\mbox{\tiny LASSO}}$ as well as the supervised estimators, thereby indicating its unsuitability as an estimator of the $\bbeta_0$ direction in this case.

These results, together with those from the simulation studies, suggest that the U$_{\mbox{\tiny LASSO}}$ estimator \emph{can achieve accuracy close to or sometimes better} than supervised algorithms trained using labeled data of sizes as large as $n=500$, which notably is close to the \emph{largest} number of labels available for most EMR phenotyping projects \citep{liao2015development}.





\section{Discussion}\label{P2:disc}

In this paper, we considered an unsupervised signal recovery problem in single index models for binary outcomes with assistance from a `surrogate' variable that exhibits surrogacy only its extreme tail regions. We proposed a simple $\ULASSO$ estimator with provable performance guarantees (under suitable assumptions) for sparse signal recovery based on an $L_1$-penalized estimation procedure in an extreme subset, defined by the surrogate variable, of the data. The problem setting we consider in this paper is fairly unique and recent, and has particular relevance in a variety of modern applications, including biomedical studies based on large databases like EMR which served as one of our primary motivations. The initial results and discussions in Section \ref{P2:basic_found} (as well as our illustrations and analyses in Appendix \ref{P2:splcase} for a simple and common subclass of models) 
also highlight the fundamental challenges underlying our problem setting and the necessity of further structural assumptions like sparsity.

We provide precise finite sample performance bounds for our estimator establishing its convergence rates, among other interesting implications, as well as illustrations through simulations
and applications to real data all of which seem to yield fairly satisfactory results.
The performance of our U$_{\mbox{\tiny LASSO}}$ estimator is quite robust to the choice of $q$ provided that $q$ itself and the corresponding $\pi_q$ are both reasonably small. This also indicates that the estimator can perhaps be further combined appropriately over multiple choices of $q$, as shown in our data example, leading to a more stable and efficient estimator. Results from both simulation studies and real applications on EMR data
demonstrate that the unsupervised algorithm trained via U$_{\mbox{\tiny LASSO}}$ achieves accuracy comparable to or higher than supervised algorithms trained on labeled datasets of \emph{most practical sizes}.

Apart from the recovery of  $\bbeta_0$, a key consequence of our sparsity based approach is its ability to perform variable selection. This can be quite useful in subsequent analyses based on an actual training data with $Y$ observed, wherein only the selected variables may be used and this can significantly improve the efficiency/accuracy of the final classification rule.

\par\medskip
\noindent \emph{Discussions on the Model Assumptions.}\label{P2:disc_sim}
Another aspect of our approach that is perhaps worth some more discussions is the SIM framework we adopt for $Y$ (and $S$) given $\bX$, thereby making the parameter $\bbeta_0$ recoverable only upto scalar multiples. While it naturally allows for more flexible models compared to standard parametric models with known link functions, it is also somewhat `necessary' under our setting, where if $S \in \Isc_q$, $Y$ is very likely to equal $Y^*_q$ which is deterministically $0/1$ in the tails of $S$. 
Hence, $Y$ itself must quickly approach a similar noiseless (and `link-free') form as $q \downarrow 0$, thus making $\bbeta_0$ (essentially) identifiable only upto scalar multiples under $\P_q(\cdot)$ even if $Y$'s link
function is known (see Remark \ref{P2:splcase_rem4} in Appendix \ref{P2:splcase} 
for further details and clarifications on this issue under a more familiar set-up).

As far as the SIM assumption \eqref{P2:s_model} for $S \given \bX$ is concerned, while it is again more flexible than standard parametric models, even more general models may also be considered given that the size of the original data $\Usc_N^*$ is very large (so that the relationship between $S$ and $\bX$ can be learnt more accurately via even more flexible non/semi-parametric models). Our main purpose in introducing the SIM for $S \given \bX$ was to underline the fundamental difficulty of the problem at hand \emph{even under} the relatively simple and interpretable model based set-up in \eqref{P2:y_model}-\eqref{P2:s_model}, as highlighted in Theorem \ref{P2:THM_1}. This also helps provide key understandings of the problem including the fundamental limitations, the connections between the $\bbeta_0$ and the $\balpha_0$ directions, and the necessary conditions for possible recovery of $\bbeta_0$ and the success of our approach. These insights are made even more explicit under a specific subclass of models for $(Y,S,\bX')'$ in Appendix \ref{P2:splcase} where many useful remarks and discussions are provided.

\par\medskip
\noindent\emph{Choice of the Squared Loss and Use of Other Loss Functions.}\label{P2:disc_oth_losses}
 As mentioned in Section \ref{P2:basic_found}, while our approach is based on the squared loss (which, for binary $Y$, is closely related to linear discriminant analysis), other convex loss functions more suited for binary outcomes, such as the logistic loss, can also be considered. However, the corresponding technical analyses can be considerably more involved, since the (unavoidable) design restriction $S \in \Isc_q$ with $q$ small, underlying our setting, can correspond to highly flat regions in the curves of these other more traditional choices of loss functions (see Figure \ref{P2:surrogacy_plot}(a) in Appendix \ref{P2:splcase} 
for an illustration), and their minimization can lead to potential non-identifiability issues. 
The squared loss on the other hand does not have any such issues as long as $\bSigma_q \succ 0$. We do therefore believe that the squared loss is a somewhat
safer and more convenient choice under our setting. Further, a variety of recent works in high dimensional statistics and compressed sensing have also generally relied (implicitly or explicitly) on use of the squared loss for signal recovery in SIMs, including \citet{Thramp_2015, Neykov_L1_2016, Plan_2016, Plan_2017, Wei_2018}, among several others. Most of these works also consider the case of binary outcomes and/or one-bit compressed sensing as a special case of their general framework; see also \citet{Plan_Vershynin_2013_a} which focusses specifically on binary outcomes and considers sparse signal recovery in SIMs based on an approach similar (although not equivalent) to a squared loss based optimization.

\par\medskip
\noindent\emph{Possible Extensions of Our Approach.}\label{P2:disc_extns}
While we have focussed on standard $L_1$ penalization throughout for simplicity, other sparsity friendly penalties like the ALASSO penalty can also be considered. Moreover, we have focussed here on a setting with a single surrogate. The proposed procedure can also be extended to settings where we have multiple such surrogates available, each satisfying the desired assumptions, in which case the estimators of the $\bbeta_0$ direction obtained from each of them (and possibly over several choices of $q$) can be further combined effectively to give a more stable and efficient estimator. Lastly, we have throughout focussed on the estimation of $\bbeta_0$ and not quite pursued any inference based on our $\ULASSO$ estimator. Considering that it is essentially a simple least squares LASSO estimator, inference based on the $\ULASSO$ estimator can indeed be performed using a de-biased LASSO type approach \citep{VdG_2014, Javanmard_2014} appropriately adapted to our setting. However, given the fundamental challenges underlying even the estimation problem, we have refrained from venturing into this topic in this paper and leave it for future research.

\par\medskip
\noindent\emph{Other Related Areas of Work.} Finally, it is also worth pointing out that our work is loosely related to some recent works in one-bit compressed sensing with `adversarial' bit flips or corruptions \citep{Plan_Vershynin_2013_a, Genzel_2017},
as well as more classical works in statistics on measurement error and misclassification \citep{Carroll_2006, Buonaccorsi_2010} and in machine learning on classification with imperfect labels \citep{Natarajan_2013,Frenay_2014}. However, in all these other lines of literature, the problem setting as well as the approach and the associated assumptions adopted are quite different, and their connections to our work are remote at best.
One fundamental difference in the basic framework for these problems compared to ours is that a proper surrogate outcome (typically, a systematically noisy or randomly flipped version of the true outcome) is \emph{actually} observed for \emph{all} individuals, which is quite unlike the setting we consider here. Under our setting, we don't `observe' $Y$ (noisy or not) at all and instead, we use the surrogacy of $S$ to `synthesize' our outcomes in the tails of $S$, where the surrogate outcome can be somewhat `trusted'. The extreme subset $\Usc_{\nq}$ in (\ref{P2:restricted_surr_data}), for some small enough $q$, provides us with the only (and that too approximate) access to the corresponding true $Y$. This is only a very small fraction of the entire observed data $\Usc_N^*$ and furthermore, the surrogate outcome therein cannot be simply treated as a randomly flipped version of the original outcome.

\appendix
\numberwithin{equation}{section}




\section{Assumptions for Theorems \ref{P2:THM_2} and \ref{P2:COR_1}}\label{P2:assumption}

In this section, we list a few key technical assumptions, along with discussions on their applicability, required for our main results regarding the finite sample performance guarantees for our proposed $\ULASSO$ estimator. 
We first discuss Assumptions \ref{P2:strngconv_assmpn} and \ref{P2:sparsity_assmpn} required for Theorem \ref{P2:THM_2}, followed by Assumption \ref{P2:subgaussian_defn} required for Theorem \ref{P2:COR_1}.

\begin{assumption}[Restricted Strong Convexity]\label{P2:strngconv_assmpn}
\emph{We assume that at $\bbeta = \bbetabar_q$, the loss function $\Lsc_{\nq}(\Usc_{\nq};\bbeta)$ in (\ref{P2:esql}) satisfies a restricted strong convexity property as follows: there exists a (non-random) constant $\kappa_q > 0$, possibly depending on $q$, such that
\begin{equation}
\nabla_2\{\Lsc_{\nq}(\Usc_{\nq};\bbetabar_q; \bv)\}  \; \equiv \; \bv'\widehat{\bSigma}_q \bv \; \geq \; \kappa_q \|\bv\|_2^2 \quad \forall \; \bv \in \C(\bbeta_0;\bbetabar_q), \; \mbox{where} \label{P2:strngconv_eqn}
\end{equation}
\begin{equation*}
\nonumber \C(\bbeta_0;\bbetabar_q)  \; = \; \{ \bv \in \R^p: \| \Pi_{\bbeta_0}^c(\bv)\|_1 \leq 3 \| \Pi_{\bbeta_0}(\bv)\|_1 + 4 \| \Pi_{\bbeta_0}^c(\bbetabar_q) \|_1\} \; \subseteq \R^p. \label{P2:Cset_def}
\end{equation*}
}
\end{assumption}
Assumption \ref{P2:strngconv_assmpn}, largely adopted from \citet{Negahban_2012}, is one of the several restricted eigenvalue type assumptions that are standard in the high dimensional statistics literature. While we assume (\ref{P2:strngconv_eqn}) to hold deterministically for any realization of $\Usc_{\nq}$, it only needs to hold almost surely 
$[\P_q]$ for some $\kappa_q$. With appropriate modifications, it can also be generalized further, wherein it only needs to hold with high probability. 
In general, if $\bSigma_q \succ 0$ and $\bX \given S \in \Isc_q$ is sufficiently well behaved (for example, subgaussian), then (\ref{P2:strngconv_eqn}) can be shown to hold (see \citet{Negahban_2012}, \citet{Rud_Zhou_2013} and references therein for relevant results)
with high probability  
for some $\kappa_q \gtrsim \lambda_{\min}(\bSigma_q)$. 

We next state our second set of assumptions for Theorem \ref{P2:THM_2}, which relates to structured sparsity conditions on $\bbeta_0$, $\balpha_0$ and some arbitrary realization of $\bbetahat_{\nq}(\lambda)$.

\begin{assumption}[Restricted Sparsity Conditions]\label{P2:sparsity_assmpn}
\emph{We assume the following conditions:} 
\emph{
\begin{enumerate}
\item[C1.] \label{test:C1} $\bbeta_0$ is strictly \emph{sparser} than $\balpha_0$ in the sense that $\Asc^c(\bbeta_0) \cap \Asc(\balpha_0)$ is non-empty and hence, $\exists \; j \in \{1,\hdots,p\}$ such that $\bbeta_{0 [j]} = 0$ and $\balpha_{0 [j]} \neq 0$.
\item[C2.] The tuning parameter $\lambda$ is $(\bbeta_0, \balpha_0, q)-$admissible, defined as follows: $\exists$ \emph{some} realization $u_{\nq}$ (\emph{not} necessarily the observed one) of the data $\Usc_{\nq}$ such that the corresponding estimator $\bbetahat_{\nq}(\lambda; u_{\nq})$, based on $u_{\nq}$ and the given choice of $\lambda$, satisfies the property: $\bbetahat_{\nq [j]}(\lambda; u_{\nq}) = 0$ for some $j \in \Asc^c(\bbeta_0) \cap \Asc(\balpha_0)$ (a non-empty set under condition C1).
\end{enumerate}
}
\end{assumption}
Assumption \ref{P2:sparsity_assmpn} imposes some mild restrictions on the sparsity patterns of $\bbeta_0$, $\balpha_0$, and one arbitrary realization of $\bbetahat_{\nq}(\lambda; \Usc_{\nq})$ for a given choice of $\lambda$. Condition C1 ensures that $\bbeta_0$ is among one of the favorable sparser directions. Specifically, it implies that $\bbeta_0$ is sparser than $\balpha_0$ in at least one coordinate, and therefore, formally characterizes the very essence  behind our consideration of a penalized regression approach for recovering $\bbeta_0$, which is based on the intuition that the penalized solution would be favoring sparser solutions and try to push it away from the un-regularized solution that recovers the $\balpha_0$ direction. Condition C2 is a very mild, yet critical, assumption regarding the sparsity structure of the solution $\bbetahat_{\nq}(\lambda; \Usc_{\nq})$ for at least \emph{one} realization of $\Usc_{\nq}$ over the entire sample space underlying the generation of $\Usc_{\nq}$. Given a $\lambda$, it requires the estimator, at least for one arbitrary sample point, to be sparse in one of the coordinates of $\Asc^c(\bbeta_0) \cap \Asc(\balpha_0)$. Of course, while this can be ensured by making $\lambda$ large, the main utility of this condition lies in ensuring that even for choices of $\lambda$ of a reasonably small enough order, at least one of the coordinates of $\balpha_0$ in 
$\Asc^c(\bbeta_0) \cap \Asc(\balpha_0)$ should not be too large so as to be always selected by $\bbetahat_{\nq}(\lambda; \Usc_{\nq})$. 

\subsection{Assumptions for Theorem \ref{P2:COR_1}}\label{P2:subsec_assumption}

We next state the conditions required for Theorem \ref{P2:COR_1} (as well as its generalized version Theorem \ref{P2:THM_3} given later in Appendix \ref{P2:thm3.2_proof}), wherein we require the distribution of $\bX \given S \in \Isc_q$ to be subgaussian so that it has sufficiently well behaved tails, defined as follows.

\begin{assumption}[Subgaussian Distributions Conditional on $S \in \Isc_q$]\label{P2:subgaussian_defn}
\emph{We assume that $\bX \mid S \in \Isc_q$ follows a $S \in \Isc_q$-restricted subgaussian distribution, defined as follows. Let $Z \in \R$ and $\bZ \in \R^p$ be any scalar and vector valued measurable functions of $(Y,S,\bX')'$ respectively. Let $\widetilde{Z}_q = Z - \E_q(Z)$ and $\widetilde{\bZ}_q = \bZ - \E_q(\bZ)$ denote their corresponding centered versions given $\{S \in \Isc_q\}$, for any $q \in (0,1]$. Then, $Z$ is said to follow a \emph{$\{S \in \Isc_q\}$-restricted subgaussian distribution} with parameter $\sigma_q^2$ for some constant $\sigma_q > 0$, denoted as:
\begin{equation}
Z \sim \SGq(\sigma_q^2), \quad \mbox{if} \;\; \E_q\{\exp (t \widetilde{Z}_q)\} \; \leq \; \exp(\sigma_q^2 t^2 /2) \quad \forall \; t \in \R. \label{P2:SGq:def}
\end{equation}
Further, a $p$-dimensional $\bZ$ is said to follow a $\{S \in \Isc_q\}$-restricted subgaussian distribution, with parameter $\sigma_q^2$,  denoted as $\bZ \sim \SGq(\sigma_q^2),$ if for each $\mathbf{t} \in \R^p$, $\mathbf{t}'\bZ$ follows a $\{S \in \Isc_q\}$-restricted subgaussian distribution with parameter at most $\sigma_q^2 \| \mathbf{t} \|^2_2$ for some constant $\sigma_q > 0$.}
\end{assumption}
The conditions in Assumption \ref{P2:subgaussian_defn} are quite mild, and should be expected to hold for a fairly large family of distributions for $(S,\bX')'$, especially those where the unconditional distribution of $(S,\bX')'$ is itself subgaussian. In particular, it holds trivially for any $q$ if $\bX$ is assumed to be bounded. Further, when $(S,\bX')'$ follows a multivariate normal distribution,
it can be shown (see result (iv) in Theorem \ref{P2:thm_4} given in Appendix \ref{P2:splcase}) 
that for most small enough $q$ of our interest (in fact, for any $q \leq 1/2$), $\bX \sim \SGq (\sigma_q^2)$ indeed with $\sigma_q^2 \leq c_1 \deltabar_q^2 \leq c_2 \log(q^{-1})$ for some constants $c_1, c_2 > 0$. Note that the parameter $\sigma_q$ in (\ref{P2:SGq:def}) is, in general, allowed to depend on $q$ and therefore, possibly diverge (slowly enough) as $q$ decreases.

Finally, it is also worth noting that a closer inspection of the proofs of Theorems \ref{P2:COR_1} and \ref{P2:THM_3}, where Assumption \ref{P2:subgaussian_defn} is most needed, will reveal that the conditions of Assumption \ref{P2:subgaussian_defn} can be somewhat weakened without altering the results. Instead of assuming a joint subgaussianity of the full vector $\bX \given S \in \Isc_q$, it may suffice to require only marginal (coordinate-wise) subgaussianity of $\bX \given S \in \Isc_q$ uniformly in the coordinates. Specifically, the results continue to hold if one only assumes that $\bX_{[j]} \sim \SGq(\sigma_q^2)$ uniformly in $j \in \{1,\hdots, p\}$ for some constant $\sigma_q \geq 0$ (independent of $j$). This is clearly weaker than the joint subgaussianity required in Assumption  \ref{P2:subgaussian_defn} and includes, for instance, the case where $\bX$ has uniformly bounded coordinates. Nevertheless, we shall stick to the slightly stronger (and more standard) formulation of joint subgaussianity, as in Assumption \ref{P2:subgaussian_defn}, for the sake of simplicity.

\section{Supporting Lemmas for the Main Proofs}\label{P2:prelim_lemmas}
We state here a few preliminary lemmas that would be useful throughout in the proofs of several of the main theorems, including in particular, Theorem \ref{P2:COR_1} and Theorem \ref{P2:THM_3}, as well as Theorem \ref{P2:thm_4} introduced later in Appendix \ref{P2:splcase}.

\begin{lemma}[Properties of Subgaussian Variables]\label{P2:lemma_1}
Let $Z$ be any random variable such that $\E(Z) = 0$ and $Z$ follows a subgaussian distribution with parameter $\sigma^2$ for some $\sigma \geq 0$, to be denoted as $Z \sim \S\G (\sigma^2)$, so that $\E\{\exp(aZ)\} \leq \exp(\sigma^2 a^2 /2) \;\; \forall \; a \in \R$. Then,
\begin{enumerate}[(i)]
\item For any $\epsilon > 0$, $\P(Z > \epsilon) \hspace{0.015in} \leq \hspace{0.015in} \exp \left\{-\epsilon^2/(2\sigma^2) \right\}$ and $\; \P(| Z | > \epsilon) \hspace{0.015in} \leq \hspace{0.015in} 2  \exp \left\{- \epsilon^2/(2\sigma^2) \right\}$.
\item $aZ \sim \S\G (a^2 \sigma^2)$ $\forall \; a \in \R$. For any $Z_1 \sim \S\G (\sigma_1^2)$ and $Z_2 \sim \S\G(\sigma_2^2)$, with $Z_1$ and $Z_2$ not necessarily independent, $(Z_1 + Z_2) \sim \S\G \{ (\sigma_1 + \sigma_2)^2\}$. If $Z_1$ and $Z_2$ are further independent, then $(Z_1 + Z_2) \sim \S\G (\sigma_1^2 + \sigma_2^2)$, with an improved parameter.
\item For each integer $m \geq 2$, $\E(|Z|^m) \leq  2 (\sqrt{2}\sigma)^m \Gamma(m/2 + 1)$, where $\Gamma(\cdot)$ denotes the gamma function: $\Gamma(t) = \int_0^\infty x^{t-1} \exp(-x) \; dx \;\; \forall \; t \geq 0$.
\item For any collection $\{ Z_j\}_{j=1}^m$ of $m$ random variables, each subgaussian with parameter $\sigma^2$, $\P\left(\underset{1 \leq j \leq m}{\max}\; |Z_j|  > \epsilon \right) \hspace{0.015in} \leq \hspace{0.015in} 2  \exp\left\{-\epsilon^2/(2\sigma^2) + \log  m \right\}$ for any $\epsilon > 0$.
\item A random vector $\bZ \in \R^d$ for any $d$, with $\E(\bZ) = \bzero$, is said to follow a subgaussian distribution with parameter $\sigma^2$ for some $\sigma \geq 0$, denoted as $\bZ \sim \S\G(\sigma^2)$, if $\forall \; \bt \in \R^d$, $\bt'\bZ \sim \S\G\{\sigma^2(t)\}$ for some $\sigma(t) \geq 0$ such that $\sigma^2(t) \leq  \sigma^2 \| \bt\|_2^2 $.
Let $\{ \bZ_j\}_{j=1}^m$ be any collection of $m$ random vectors (not necessarily independent) in $\R^d$ with $\bZ_j \sim \S\G(\sigma^2)$ $\forall \; 1 \leq j \leq m$. Then, $\P\left(\underset{1 \leq j \leq m}{\max}\; \|\bZ_j\|_{\infty}  > \epsilon \right) \hspace{0.015in} \leq \hspace{0.015in} 2 \exp\left\{-\epsilon^2 /(2\sigma^2) + \log  (md) \right\}$ $\; \forall  \; \epsilon > 0$.
\end{enumerate}
\end{lemma}

\begin{lemma}[Subgaussian Properties of Binary Random Variables]\label{P2:lemma_2}
Let $Z \in \{0,1\}$ be a binary random variable with $\E(Z) \equiv \P(Z = 1) = p \in [0,1]$. Let $\widetilde{Z} = (Z-p)$ denote the corresponding centered version of $Z$. Then, $\widetilde{Z} \sim \S\G(\widetilde{p}^{\;2})$, where $\widetilde{p} \geq 0$ is given by: $\widetilde{p} = 0$ if $p \in \{0,1\}$, $\widetilde{p} = 1/2$ if $p = 1/2$, and $\widetilde{p} = [(p-1/2)/\log \{p/(1-p)\}]^{1/2}$ if $p \notin \{0,1,1/2\}$.
\end{lemma}

\begin{lemma}[Bernstein's Inequality]\label{P2:lemma_3}
Let $\{\Z_1, \hdots, Z_n \}$ denote any collection of $n$ independent (not necessarily iid) random variables $\in \R$, such that $\E(Z_i) = 0 \; \forall \; 1 \leq i \leq n$. Suppose $\exists$ constants $\sigma \geq 0$ and $K \geq 0$, such that $n^{-1} \sum_{i=1}^n \E(|Z_i|^m) \leq (m!/2) K^{m-2} \sigma^2$, for each positive integer $m \geq 2$. Then, the following concentration bound holds:
 $$ \P \left( \frac{1}{n} \left| \sum_{i=1}^n Z_i \right| \geq \sqrt{2}\sigma \epsilon + K \epsilon^2\right) \; \leq \; 2 \exp(-n\epsilon^2) \quad \mbox{for any} \;\; \epsilon > 0. $$
\end{lemma}

\begin{lemma}[Useful Bounds for the Standard Normal Density and CDF]\label{P2:lemma_4}
Let $\phi(\cdot)$ and $\Phi(\cdot)$ respectively denote the density and the CDF of the standard $\Nsc_1(0,1)$ distribution. Further, let $\Phibar(t) = \{1 -\Phi(t)\} \equiv \Phi(-t) \; \forall \; t\in \R$. Then, the following bounds hold: 
$\forall \; t > 0$,
$$ \frac{t}{1+ t^2} \phi(t) \; \leq \; \Phibar(t) \; \leq \; \frac{\phi(t)}{t}, \;\;\; \mbox{and} \;\;\; \Phibar(t) \; \leq \; \exp ( - t^2 /2 ). $$
\end{lemma}

\begin{lemma}[Properties of the Truncated Normal Distribution] \label{P2:lemma_5}
Let $Z \sim \Nsc_1(0,\sigma^2)$ for some $\sigma > 0$, and let $\phi(\cdot)$ and $\Phi(\cdot)$ respectively denote the density and the CDF of the standard $\Nsc_1(0,1)$ distribution. For any $a, b$ such that $ - \infty \leq a < b \leq \infty$, consider the truncated random variable: $Z_{a,b} \equiv (Z \given a \leq Z \leq b)$. Let $\abar = (a/\sigma)$ and $\bbar = (b/\sigma)$. Then,
\begin{enumerate}[(i)]
\item $Z_{a,b}$ satisfies the following distributional properties:
\begin{alignat*}{2}
\nonumber & \E(Z \given a \leq Z \leq b) \; && = \; \sigma \frac{\phi(\abar) - \phi(\bbar)}{\Phi(\bbar) - \Phi(\abar)}, \\
\nonumber & \E(Z^2 \given a \leq Z \leq b) \; && = \; \sigma^2\left\{ 1 + \frac{\abar \phi(\abar) - \bbar \phi(\bbar)}{\Phi(\bbar) - \Phi(\abar)} \right\}, \;\; \mbox{and} \;\; \forall \; t \in \R, \\ 
\nonumber 
& \E\left\{ \exp(tZ) \given a \leq Z \leq b \right\} \; && = \; \exp\left(\sigma^2 t^2 /2\right)\left\{\frac{\Phi(\bbar - \sigma t) - \Phi(\abar - \sigma t)}{\Phi(\bbar) - \Phi(\abar)}\right\}. 
\end{alignat*}
\item For any $q \in (0,1]$, let $z_q$ and $\zbar_q$ respectively denote the $(q/2)^{th}$ and $(1 - q/2)^{th}$ quantiles of the standard $\Nsc_1(0,1)$ distribution, so that $ - z_q = \zbar_q \geq 0$. Define: 
    $f_q(t) = \frac{1}{2 \Phi(z_q)} \{ \Phi(z_q + t) + \Phi(z_q -t)\}$ $\forall \; t \in \R$. Then, the function $f_q(\cdot)$ satisfies: for any $t \in \R$,
\begin{equation*}
f_q(t) \; \leq \; \exp \left(t^2 \zbar_q^2\right) \;\; \forall \; q \in (0,1/2] \;\;\;  \mbox{and} \;\;\; f_q(t) \; \leq \; 2 \;\;  \forall \; q \in (1/2,1].
\end{equation*}
\end{enumerate}
\end{lemma}

Lemma \ref{P2:lemma_1} is a collection of several well-known properties of subgaussian distributions, and proofs and/or discussions of these results (or equivalent versions) can be found in several relevant references, including \citet{Vershynin_2010} for instance. Lemma \ref{P2:lemma_2} explicitly characterizes the subgaussian properties of (centered) binary random variables, and its proof can be found in \citet{Buldygin_2013}. Lemma \ref{P2:lemma_3} is one of many versions of the well-known Bernstein's inequality, and this particular version has been adopted from \citet{Van_de_Geer_2013}. Lemma \ref{P2:lemma_4} provides some useful and fairly well known bounds involving the standard normal CDF and density, and their mentions and/or discussions can be found in \citet{Dumbgen_2010}, \citet{Chiani_2003} and the references cited therein. Lastly, Lemma \ref{P2:lemma_5} provides some useful distributional properties of truncated normal distributions. For the results in (i), proofs and/or mentions of them (or much more general versions) can be found in a combination of references including \citet{Tallis_1961}, \citet{Johnson_1994}, \citet{Horrace_2004, Horrace_2005} and \citet{Burkardt_2014}. Result (ii) in Lemma \ref{P2:lemma_5} is a fairly straightforward conclusion, and can be obtained, for instance, through direct numerical verification. We skip the details here for the sake of brevity, and leave them to the interested reader to verify.

\section{Proof of Theorem \ref{P2:THM_1} }\label{P2:pf_thm1}

We first note that owing to the model assumptions in (\ref{P2:y_model})-(\ref{P2:s_model}), and the linearity condition (\ref{P2:lin_condn_eqn_1}) in Assumption \ref{P2:dlc_assmpn}, we have
\begin{align}
& \E(\bv'\bX \given \balpha_0'\bX, \bbeta_0'\bX, \epsilon, \epsilon^*) 
 \; = \; c_{\bv} + a_{\bv}(\balpha_0'\bX) + b_{\bv}(\bbeta_0'\bX)  \quad \forall \; \bv, \quad \mbox{and} \label{P2:pf1_eqn1} \\
\bv'\bmu_q  & \equiv \; \E_q (\bv'\bX) 
 \; = \; \E_q \{ \E(\bv'\bX \given \balpha_0'\bX, \bbeta_0'\bX, \epsilon, \epsilon^*) \} \nonumber \\
& = \; \E_q\{c_{\bv} + a_{\bv}(\balpha_0'\bX) + b_{\bv}(\bbeta_0'\bX)\} \; = \; c_{\bv} + a_{\bv}(\balpha_0'\bmu_q) + b_{\bv}(\bbeta_0'\bmu_q) .  \label{P2:pf1_eqn2}
\end{align}
Note that for all the steps in obtaining (\ref{P2:pf1_eqn2}), it is implicitly understood, as would be the case henceforth, that the values assumed by the conditioning variables $\{\bbeta_0'\bX, \balpha_0'\bX, \epsilon, \epsilon^*\}$ are such that the underlying restriction $\{S \in \Isc_q\}$ is indeed feasible so that $\E_q(\cdot\given \bbeta_0'\bX, \balpha_0'\bX, \epsilon, \epsilon^*)$ is well-defined, and can be further replaced by $\E(\cdot\given \bbeta_0'\bX, \balpha_0'\bX, \epsilon, \epsilon^*)$, owing to (\ref{P2:s_model}), as used while obtaining (\ref{P2:pf1_eqn2}). Next, note that $\L_q(\bv)=\E_q[\{Y - p_q - \bv'(\bX-\bmu_q)\}^2]$ satisfies
\begin{eqnarray}
\nonumber \L_q(\bv) 
\nonumber &=& \E_q(\E[\{Y - p_q - \bv'(\bX-\bmu_q)\}^2 \given \bbeta_0'\bX, \balpha_0'\bX, \epsilon, \epsilon^*])\\
\nonumber &\geq& \E_q \left\{ \left(\E[\{Y - p_q - \bv'(\bX-\bmu_q)\} \given \bbeta_0'\bX, \balpha_0'\bX, \epsilon, \epsilon^*]\right)^2\right\} \\
\nonumber &=& \E_q \left[  \left\{ Y - p_q - a_{\bv}(\balpha_0'\bX) - b_{\bv}(\bbeta_0'\bX) + a_{\bv}(\balpha_0'\bmu_q) + b_{\bv}(\bbeta_0'\bmu_q)\right\}^2 \right] \\
&=& \E_q \left[ \left\{ Y - p_q - (a_{\bv}\balpha_0 + b_{\bv} \bbeta_0)'(\bX - \bmu_q)\right\}^2 \right]  \equiv  \L_q( a_{\bv}\balpha_0 + b_{\bv} \bbeta_0). \label{P2:pf1_eqn3}
\end{eqnarray}
The first equality in obtaining (\ref{P2:pf1_eqn3}) follows from arguments similar to those mentioned earlier while obtaining (\ref{P2:pf1_eqn2}). The subsequent inequality follows from (conditional) Jensen's inequality, while in the penultimate step we have used (\ref{P2:pf1_eqn1})-(\ref{P2:pf1_eqn2}), as well as the fact that owing to (\ref{P2:y_model}), $Y$ is completely determined (hence constant) by the conditioning variables $\{\bbeta_0'\bX, \balpha_0'\bX, \epsilon, \epsilon^*\}$. Thus, (\ref{P2:pf1_eqn3}) now shows that the value of $\L_q(\cdot)$ at every $\bv \in \R^p$ is bounded below by its value at a corresponding point of the form $(a_{\bv}\balpha_0 + b_{\bv}\bbeta_0) \in \R^p$. In particular, this also applies to $\bv = \bbetabar_q$, which however is the unique minimizer of $\L_q(\bv)$ over $\bv \in \R^p$. Hence, $\bbetabar_q$ must be of the form $(a_{\bv}\balpha_0 + b_{\bv}\bbeta_0)$ with $\bv = \bbetabar_q$. This establishes (\ref{P2:thm1_eqn1}). \qed

To show (\ref{P2:thm1_eqn2}), we first note that owing to (\ref{P2:s_model}) and (\ref{P2:lin_condn_eqn_1})-(\ref{P2:lin_condn_eqn_2}),  
\begin{align}
& \E(\bv'\bX \given \balpha_0'\bX, \epsilon^*) 
 \; = \; (c_{\bv} + b_{\bv} \cbar) + (a_{\bv} +  b_{\bv}\abar)(\balpha_0'\bX) \quad \forall \; \bv,  \quad \mbox{and} \label{P2:pf1_eqn4}\\
&  \bv'\bmu_q  
 \; = \; \E_q \{ \E(\bv'\bX \given \balpha_0'\bX, \epsilon^*) \}
 \; = \; (c_{\bv} + b_{\bv} \cbar ) + (a_{\bv} + b_{\bv} \abar )(\balpha_0'\bmu_q),
\label{P2:pf1_eqn5}
\end{align}
where, for the first equality in obtaining (\ref{P2:pf1_eqn5}), we  used the fact that  $S$ is completely determined by  $\{\balpha_0'\bX, \epsilon^*\}$, so that the term 
$\E_q(\cdot \given \balpha_0'\bX, \epsilon^*)$ inside can be replaced by
$\E(\cdot \given \balpha_0'\bX, \epsilon^*)$,
owing to (\ref{P2:s_model}). Next, note that  $\L_q^*(\bv) \equiv \E_q[\{Y^*_q - p_q^* - \bv'(\bX-\bmu_q)\}^2]$ satisfies
\begin{eqnarray}
\nonumber \L_q^*(\bv) 
\nonumber &=& \E_q(\E[\{Y^*_q - p_q^* - \bv'(\bX-\bmu_q)\}^2 \given \balpha_0'\bX, \epsilon^*])\\
\nonumber &\geq& \E_q \left\{ \left(\E[\{Y^*_q - p_q^* - \bv'(\bX-\bmu_q)\} \given \balpha_0'\bX, \epsilon^*]\right)^2\right\} \\
\nonumber &=& \E_q \left[ \left\{ \E(Y^*_q \given  \balpha_0'\bX, \epsilon^*) - p_q^* - \E(\bv'\bX \given \balpha_0'\bX, \epsilon^*) + \bv'\bmu_q \right\}^2 \right] \\
&=& \E_q \left[\left\{ Y^*_q - p_q^* - (a_{\bv} + b_{\bv} \abar)\balpha_0'(\bX - \bmu_q)\right\}^2 \right]   \equiv  \L_q^*\{(a_{\bv} + b_{\bv} \abar) \balpha_0\}. \label{P2:pf1_eqn6}
\end{eqnarray}
The first equality in obtaining (\ref{P2:pf1_eqn6}) follows from arguments similar to those mentioned earlier while obtaining (\ref{P2:pf1_eqn5}). The subsequent inequality follows from (conditional) Jensen's inequality, while in the penultimate step we have used (\ref{P2:pf1_eqn4})-(\ref{P2:pf1_eqn5}), as well as the fact that owing to (\ref{P2:s_model}) and the very definition of $Y^*_q$, $Y^*_q$ is completely determined (hence constant) by the conditioning variables $\{\balpha_0'\bX, \epsilon^*\}$. Thus, (\ref{P2:pf1_eqn6}) now shows that the value of $\L_q^*(\cdot)$ at every $\bv \in \R^p$ is bounded below by its value at a corresponding point of the form $(a_{\bv}+ b_{\bv}\abar)\balpha_0 \in \R^p$. In particular, this also applies to $\bv = \balphabar_q$, which however is the unique minimizer of $\L_q^*(\bv)$ over $\bv \in \R^p$. Hence, $\balphabar_q$ must be of the form $(a_{\bv} + b_{\bv}\abar)\balpha_0$ with $\bv = \balphabar_q$. We have therefore established (\ref{P2:thm1_eqn2}), as required. The proof of Theorem \ref{P2:THM_1} is complete. \qed


\section{Proof of Theorem \ref{P2:THM_2}}\label{P2:pf_thm2}

The proof of this result relies substantially on a useful result from \citet{Negahban_2012}. We will therefore try to adopt some of their basic notations and terminology at the beginning of this proof in order to facilitate the use of that result.

For any $\bu \in \R^p$, let $\Rsc(\bu) = \| \bu \|_1$ and let $\Rsc^*(\bu) \equiv 
\mbox{sup}_{\bv \in \R^p \backslash \{\bzero\}}
\{\bu'\bv/ \Rsc(\bv)\}$ be the `dual norm' for $\Rsc(\cdot)$. Further, for any subspace $\Msc \subseteq \R^p$, let $\Psi(\Msc) \equiv \mbox{sup}_{\bu \in \Msc \backslash \{\bzero \}}
\{\Rsc(\bu)/ \| \bu\|_2\}$ denote its 
`subspace compatibility constant' 
with respect to 
$\Rsc(\cdot)$. Then, with $\Jsc, \Msc_{\Jsc}$ and $\Msc^{\perp}_{\Jsc}$ as defined
in Section \ref{P2:notn_assmpn}, 
it is not difficult to show that: (i) $\Rsc(\cdot)$ is \emph{decomposable} with respect to the orthogonal subspace pair $(\Msc_{\Jsc},\Msc^{\perp}_{\Jsc})$ for any $\Jsc \subseteq  \{1,\hdots,p\}$, in the sense that $\Rsc(\bu + \bv) = \Rsc(\bu) + \Rsc(\bv)$ $\forall \; \bu \in \Msc_{\Jsc}, \bv \in \Msc^{\perp}_{\Jsc}$; (ii) $\Rsc^*(\bu) = \| \bu \|_{\infty} \; \forall \; \bu \in \R^p$; and (iii) with $\Jsc = \Asc(\bv)$ for any $\bv \in \R^p$, $\Psi^2(\Msc_{\Jsc}) = s_{\bv}$. (We refer 
to \citet{Negahban_2012} for further discussions and/or proofs of these facts). Lastly, let $P_{\Jsc}(\bv)$ and $P_{\Jsc}^{\perp} (\bv)$ respectively denote the orthogonal projections of any $\bv \in \R^p$ onto $\Msc_{\Jsc}$ and $\Msc_{\Jsc}^{\perp}$, for any $\Jsc$ as above.

Then, using the decomposability of $\Rsc(\cdot)$ over $(\Msc_{\Jsc},\Msc^{\perp}_{\Jsc})$ with $\Jsc$ chosen to be $\Asc(\bbeta_0)$ in particular, and the restricted strong convexity of $\Lsc_{\nq}(\Usc_{\nq};\bbeta)$ at $\bbeta = \bbetabar_q$, under Assumption \ref{P2:strngconv_assmpn}, we have by Theorem 1 of \citet{Negahban_2012}, that for any given $\Usc_{\nq}$ and any choice of a corresponding $\lambda$ such that $\lambda \geq 4 \|\T_{\nq}\|_{\infty}$,
\begin{eqnarray}
\left\| \bbetahat_{\nq}(\lambda;\Usc_{\nq}) - \bbetabar_q \right\|_2^2 &\leq & 9 s_{\bbeta_0} \frac{\lambda^2}{\kappa_q^2}  + 4 \left| a_q \right| \frac{\lambda}{\kappa_q} \left\| \Pi_{\bbeta_0}^c(\balpha_0)\right\|_1, \label{P2:pf2_eqn1}
\end{eqnarray}
where, while applying the result from \citet{Negahban_2012}, we have chosen the parameter $\btheta^*$, in their notation, as $\btheta^* = \bbetabar_q$, and we have also used: $2 \Rsc^*[ \boldsymbol{\nabla}\{\Lsc_{\nq}(\Znq^*;\bbetabar_q)\}] \equiv 4 \|\T_{\nq} \|_{\infty}$ as well as $P_{\Asc(\bbeta_0)}^{\perp} (\bbetabar_q) = \Pi_{\Asc^c(\bbeta_0)}(\bbetabar_q) \equiv \Pi^c_{\bbeta_0}(\bbetabar_q) = \Pi^c_{\bbeta_0}(a_q \balpha_0)$, so that $$\Rsc\{P_{\Asc(\bbeta_0)}^{\perp} (\bbetabar_q)\} \; = \; |a_q| \| \Pi_{\bbeta_0}^c(\balpha_0)\|_1.$$
Under Assumption \ref{P2:sparsity_assmpn} conditions C1-C2, with $\lambda$ chosen as in (\ref{P2:pf2_eqn1}) being further assumed to be $(\bbeta_0,\balpha_0,q)$-admissible, $\exists$ some realization $u_{\nq}$ of $\Usc_{\nq}$ such that the corresponding estimator $\bbetahat_{\nq}(\lambda; u_{\nq})$ based on $u_{\nq}$ and the given choice of $\lambda$, satisfies the property: $\bbetahat_{\nq [j]}(\lambda; u_{\nq}) = 0$ for some $j \in \Asc^c(\bbeta_0) \cap \Asc(\balpha_0)$. Noting that the bound in (\ref{P2:pf2_eqn1}) is deterministic and applies to any realization of $\Usc_{\nq}$, including $u_{\nq}$ in particular, we then have:
\begin{eqnarray*}
\nonumber \left| a_q \right|^2 | \balpha_{0 [j]} |^2 &\equiv& \left\{ \bbetahat_{\nq [j]}(\lambda;u_{\nq}) - \bbetabar_{q [j]} \right\}^2  \;\; \leq \;\; \left\| \bbetahat_{\nq}(\lambda;u_{\nq}) - \bbetabar_q \right\|_2^2 \\
\nonumber &\leq& 9 s_{\bbeta_0} \frac{\lambda^2}{\kappa_q^2}  + 4 \left| a_q \right| \frac{\lambda}{\kappa_q} \left\| \Pi_{\bbeta_0}^c(\balpha_0)\right\|_1 \quad\quad \mbox{[using (\ref{P2:pf2_eqn1})]}, 
\end{eqnarray*}
and therefore,
\begin{eqnarray}
\nonumber \left| a_q \right|  &\leq& \frac{\lambda}{\kappa_q \left|\balpha_{0[j]}\right|^2} \left[ 2 \left\|\Pi_{\bbeta_0}^c(\balpha_0)\right\|_1 +
 \left\{4\left\|\Pi_{\bbeta_0}^c(\balpha_0)\right\|_1^2 + 9 s_{\bbeta_0} \left|\balpha_{0[j]}\right|^2 \right\}^{\half}\right] \\
 \nonumber &\leq& \frac{\lambda}{\kappa_q \hspace{0.015in} \Cmin^2(\balpha_0,\bbeta_0)} \left\{ 4 \left\|\Pi_{\bbeta_0}^c(\balpha_0)\right\|_1  + 3 s_{\bbeta_0}^{\half} \Cmax(\balpha_0,\bbeta_0) \right\}  \\
&\equiv&  \frac{\lambda}{\kappa_q}\dbar(\balpha_0,\bbeta_0), \label{P2:pf2_eqn2}
\end{eqnarray}
where 
the preliminary bound on $|a_q|$ in the third last step follows from noting that the previous step leads to a quadratic inequality in $|a_q|$ and therefore, some straightforward algebra involving standard theory of quadratic inequalities leads to this bound, as well as the final bound in (\ref{P2:pf2_eqn2}). It now follows that
\begin{eqnarray}
\nonumber && \left\| \bbetahat_{\nq}(\lambda; \Usc_{\nq}) - b_q \beta_0 \right\|_2  
\;\; \leq \; \left\| \bbetahat_{\nq}(\lambda; \Usc_{\nq}) - \bbetabar_q \right\|_2 + \left|a_q \right| \left\| \balpha_0 \right\|_2 \\
\nonumber &&  \quad\;\; \leq \; \left\{ 9 s_{\bbeta_0} \frac{\lambda^2}{\kappa_q^2}  + 4 \left| a_q \right| \frac{\lambda}{\kappa_q} \left\| \Pi_{\bbeta_0}^c(\balpha_0)\right\|_1 \right\}^{\half} + \left|a_q \right| \left\| \balpha_0 \right\|_2 \\
\nonumber && \quad\;\; \leq \; \frac{\lambda}{\kappa_q}\left\{ 9 s_{\bbeta_0} + 4 \left\| \Pi_{\bbeta_0}^c(\balpha_0)\right\|_1 \dbar(\balpha_0,\bbeta_0)\right\}^{\half} +
\frac{\lambda}{\kappa_q} \dbar(\balpha_0,\bbeta_0) \left\| \balpha_0 \right\|_2 \\
&& \quad\;\; \equiv \; \frac{\lambda}{\kappa_q} \left[ \left\{ 9s_{\bbeta_0} + d_1(\balpha_0,\bbeta_0) \right\}^{\half} + d_2(\balpha_0,\bbeta_0)\right]. \label{P2:pf2_eqn3}
\end{eqnarray}
This completes the proof of Theorem \ref{P2:THM_2}. \qed

\subsection{Remarks on the Proof Technique and the Results Used}\label{P2:subsec_NegahbanRemark}
It needs to be mentioned that the result (Theorem 1) from \citet{Negahban_2012} used in our proof of Theorem \ref{P2:THM_2} is quite a powerful one. It provides highly flexible and general (deterministic) bounds for penalized $M$-estimators based on loss functions satisfying some restricted strong convexity condition (like the one we consider in Assumption \ref{P2:strngconv_assmpn}) and regularizers based on norms that are `decomposable' over orthogonal subspace pairs, as shown to hold for the $L_1$ norm for subspace pairs of the form: $(\Msc_{\Jsc},\Msc^{\perp}_{\Jsc})$ for any $\Jsc \subseteq \{ 1, \hdots, p\}$. The bounds hold for any such subspace pair which, in our case, we choose to be $\{\Msc_{\Asc(\bbeta_0)},\Msc^{\perp}_{\Asc(\bbeta_0)}\}$. More importantly, the result provides deviation bounds of the estimator with respect to \emph{any} point that can be reasonably viewed as a possible `target', and \emph{not} necessarily the exact parameter that minimizes the expected loss. Only the lower bound for $\lambda$ needs to be appropriately defined for each such `target'. Of course, the deviation bound depends directly on $\lambda$, and is only useful if the (random) lower bound, defined by this `target', for $\lambda$ can be bounded above with high probability (w.h.p.) by a sequence converging fast enough to $0$. In our case, owing to Theorem \ref{P2:THM_1}, $\balphabar_q$, the minimizer of $\L_q^*(\cdot)$, is really the `official' target parameter. However, even with $\bbetabar_q$ as the `target', the corresponding deviation bound for $\bbetahat_{\nq}(\lambda)$ may still satisfy reasonable rates, since the lower bound $4 \|\T_{\nq}\|_{\infty}$ for $\lambda$, defined through $\bbetabar_q$, may \emph{still} have a fast enough convergence rate w.h.p. owing to the definition of $\bbetabar_q$ in \eqref{P2:y_sqloss} and the extreme tail surrogacy of $S$ as in Assumption \ref{P2:surr_assmpn}. In fact, this is precisely what we establish in Theorem \ref{P2:THM_3} later. Nevertheless, the main purpose of this discussion was to provide some helpful insights regarding the nuances underlying our result and its proof, and also elaborate to some extent on the general usefulness of the tools used here.




\section{Proofs of Theorem \ref{P2:COR_1} and Theorem \ref{P2:THM_3}}\label{P2:thm3.2_proof}

To prove Theorem \ref{P2:COR_1}, we first note that the following decomposition holds for $\T_{\nq}$.
\begin{eqnarray}
\T_{\nq} &\equiv& \bT_{\nq}(\bbetabar_q) \;\; = \; \Tnqo + \Tnqto - \Tnqtt, \quad \mbox{where} \label{P2:tnq_def} \\
\Tnqo &=& \frac{1}{\nq} \sum_{i=1}^{\nq} (Y^*_{q,i} - Y_i)(\bX_i - \bXbar_{\nq}), \label{P2:tnq1_def}\\
 \Tnqto &=& \frac{1}{\nq} \sum_{i=1}^{\nq} \{ Y_i - p_q - \bbetabar_q'(\bX_i - \bmu_q)\} (\bX_i - \bmu_q), \quad \mbox{and} \label{P2:tnq21_def}\\
\Tnqtt &=& \frac{1}{\nq} \sum_{i=1}^{\nq} \{Y_i - p_q - \bbetabar_q'(\bX_i - \bmu_q)\} (\bXbar_{\nq} - \bmu_q) . \label{P2:tnq22_def}
\end{eqnarray}
In Theorem \ref{P2:THM_3} below, we control $\| \T_{\nq} \|_{\infty}$ by controlling the three terms
$\|\Tnqo\|_{\infty}$, $\|\Tnqto\|_{\infty}$ and $\|\Tnqtt\|_{\infty}$ individually, based on general and explicit tail bounds for each. We first state and prove Theorem \ref{P2:THM_3} and thereafter, complete the proof of Theorem \ref{P2:COR_1} as a consequence.

\begin{theorem}\label{P2:THM_3}
Suppose $\bX \sim \SGq (\sigma_q^2)$, as defined in Assumption \ref{P2:subgaussian_defn}, for some constant $\sigma_q > 0$ allowed to depend on $q$. For any $a \in [0,1]$, define $\widetilde{a} \geq 0$ as: $\widetilde{a} = 0$ if $a \in \{0,1\}$, $\widetilde{a} = 1/2$ if $a = 1/2$, and $\widetilde{a} = [(a-1/2)/\log\{a/(1-a)\}]^{1/2}$ if $a \notin \{0,1,1/2\}$. Let $\pqtil$ and $\piqtil$ denote $\widetilde{a}$ for $a = p_q$ and $a = \pi_q$ respectively. Further, let $\gamma_q^2 = (\pqtil + \sigma_q \| \bbetabar_q \|_2)^2$. 
Then, for any given $\epsilon_1, \epsilon_2, \epsilon_3, \epsilon_4, \epsilon_5 >0$,
\begin{align}
&
(i) \;\; \P_q\left\{\left\|\T_{\nq}^{(1)}\right\|_{\infty} > \epsilon_1(\pi_q + \epsilon_2)\right\} \; \leq \; 2 \exp\left\{-\frac{\epsilon_1^2}{2\sigma_q^2} + \log (\nq p)\right\} + \exp \left( - \frac{\nq \epsilon_2^2}{2 \piqtil^2} \right), \nonumber \\ 
& (ii) \; \P_q\left\{\left\|\T_{\nq, 1}^{(2)}\right\|_{\infty} > 2 \sigma_q \gamma_q \left(2\sqrt{2}\epsilon_3 +  \epsilon_3^2\right) \right\} \; \leq \; 2 \exp \left(-\nq \epsilon_3^2 + \log p\right), \;\; \mbox{and} \nonumber \\
& (iii) \; \P_q\left( \left\|\T_{\nq, 2}^{(2)}\right\|_{\infty} > \epsilon_4 \epsilon_5 \right) \; \leq  \; 2 \exp \left(-\frac{\nq\epsilon_4^2}{2 \sigma_q^2} + \log p \right) +
2 \exp \left(-\frac{\nq\epsilon_5^2}{2 \gamma_q^2}\right). \label{P2:thm3_probbound_1}
\end{align}
\end{theorem}


\begin{subsection}{Proof of Theorem \ref{P2:THM_3}}\label{P2:pf_thm3}

Since $\bX \sim \S\G_q(\sigma_q^2)$, $(\bX_i - \bXbar_{\nq}) \sim \S\G_q(\overline{\sigma}_q^2)$, owing to Lemma \ref{P2:lemma_1} (ii) and (v), where $\overline{\sigma}_q^2 = \sigma_q^2(1 - \nq^{-1})$. Then, $\Z^*_q\equiv |Y - Y^*_q |$ is a binary variable with $\P_q(\Z^*_q = 1) = \pi_q$. Hence, using Lemmas \ref{P2:lemma_2} and  \ref{P2:lemma_1} (ii), $(\Z^*_q - \pi_q) \sim \S\G_q(\piqtil^2)$ and $\nq^{-1} \sum_{i=1}^{\nq} (\Z^*_{q,i} - \pi_q) \sim \S\G_q(\piqtil^2/\nq)$. Using Lemma \ref{P2:lemma_1} (v) and (i), we therefore have: $\forall \; \epsilon_1, \epsilon_2 > 0$,
\begin{eqnarray}
&& \;\;\; \P_q \left(\underset{1 \leq i \leq \nq}{\max}\; \left\|\bX_i -\bXbar_{\nq} \right\|_{\infty}  > \epsilon_1 \right) \; \leq \; 2 \exp\left\{-\frac{\epsilon_1^2}{2\sigma_q^2} + \log  (\nq p)\right\}, \;\; \mbox{and} \label{P2:pf3_eqn1}\\
&& \;\;\; \P_q \left\{ \frac{1}{\nq}\sum_{i=1}^{\nq} \Z^*_{q,i} > (\pi_q + \epsilon_2)\right\} \; \leq \; \exp \left(- \frac{\nq\epsilon_2^2}{ 2 \piqtil^2} \right). \label{P2:pf3_eqn2}
\end{eqnarray}
Using (\ref{P2:pf3_eqn1})-(\ref{P2:pf3_eqn2}), 
we then have: $\forall \; \epsilon_1, \epsilon_2 >0$,
\begin{eqnarray}
&& \nonumber \P_q \left\{ \left\| \Tnqo \right\|_{\infty} > \epsilon_1(\pi_q + \epsilon_2) \right\} \\
&& \nonumber  \quad\;\; \equiv \; \P_q \left\{ \left\| \frac{1}{\nq} \sum_{i=1}^{\nq} (\bX_i - \bXbar_{\nq}) (Y^*_{q,i} - Y_i)\right\|_{\infty}  > \epsilon_1(\pi_q + \epsilon_2) \right\} \\
\nonumber && \quad\;\; \leq \; \P_q \left\{ \left( \underset{1 \leq i \leq \nq}{\max}\; \left\|\bX_i -\bXbar_{\nq} \right\|_{\infty} \right)
\left( \frac{1}{\nq}\sum_{i=1}^{\nq} \Z^*_{q,i}\right) > \epsilon_1(\pi_q + \epsilon_2) \right\} \\
\nonumber && \quad\;\; \leq \; \P_q \left(\underset{1 \leq i \leq \nq}{\max}\; \left\|\bX_i -\bXbar_{\nq} \right\|_{\infty}  > \epsilon_1 \right)  +  \P_q \left\{ \frac{1}{\nq}\sum_{i=1}^{\nq} \Z^*_{q,i} > (\pi_q + \epsilon_2)\right\} \\
&& \quad\;\; \leq \; 2 \exp\left\{-\frac{\epsilon_1^2}{2\sigma_q^2} + \log (\nq p) \right\} + \; \exp\left(- \frac{\nq \epsilon_2^2}{2 \piqtil^2} \right). \label{P2:pf3_eqn3}
\end{eqnarray}
(\ref{P2:pf3_eqn3}) therefore establishes the first of the three bounds in (\ref{P2:thm3_probbound_1}). \qed

Next, let us define: $\bXtil_q = (\bX - \bmu_q)$, $\Ytil_q = (Y - p_q)$ and $\Ztil_q = (\Ytil_q - \bbetabar_q'\bXtil_q)$. Then, $\bXtil_q \sim \S\G_q(\sigma_q^2)$ by assumption, $\Ytil_q \sim \S\G_q(\pqtil^2)$ owing to Lemma \ref{P2:lemma_2}, and $\Ztil_q \sim \S\G_q(\gamma_q^2)$ owing to Lemma \ref{P2:lemma_1} (ii) and (v). Hence, applying Lemma \ref{P2:lemma_1} (iii) to $\Ztil_q$ and $\bXtil_q$, we then have: for each $j \in \{1, \hdots, p\}$ and any integer $m \geq 2$, 
\begin{eqnarray*}
\nonumber \E_q \left(| \bXtil_{q [j]} \Ztil_q |^m \right) &\leq & \left\{ \E_q \left(| \bXtil_{q [j]} | ^{2m} \right) \right\}^{\half} \left\{ \E_q \left( |\Ztil_q |^{2m}\right) \right\}^{\half} \\
\nonumber & \leq& \left\{2^{m+1}\sigma_q ^{2m} \Gamma(m + 1)\right\}^{\half} \left\{ 2^{m+1} \gamma_q^{2m} \Gamma(m + 1)\right\}^{\half} \\
\nonumber & = & 2 \{\Gamma(m+1)\} (2 \sigma_q \gamma_q)^m  \; = \; \frac{m !}{2} (2 \sigma_q \gamma_q)^{m-2} (4 \sigma_q \gamma_q)^2,
\end{eqnarray*}
where the first inequality follows from Holder's inequality, and the rest are due to Lemma \ref{P2:lemma_1} (iii) applied to $\Ztil_q$ and $\bXtil_{q [j]}$. Next, note that  $\E_q(\bXtil_q \Ztil_q) = \bzero$ by definition of $\bbetabar_q$ and further, the above bound ensures that  $\bXtil_{q [j]} \Ztil_q$ satisfies the moment conditions required in Lemma \ref{P2:lemma_3} with $\sigma \equiv 4 \sigma_q \gamma_q$ and $K \equiv 2 \sigma_q \gamma_q$. Hence applying Lemma \ref{P2:lemma_3}, we have: $\forall \; \epsilon_3 > 0$, 
\begin{eqnarray}
\nonumber && \P_q\left\{\left\|\T_{\nq, 1}^{(2)}\right\|_{\infty} > 2 \sigma_q \gamma_q (2\sqrt{2}\epsilon_3 +  \epsilon_3^2) \right\} \\
\nonumber && \quad\;\; \equiv \;\; \P_q\left\{\left\| \frac{1}{\nq} \sum_{i=1}^{\nq} \bXtil_{q,i} \Ztil_{q,i}\right\|_{\infty} > 2 \sigma_q \gamma_q (2\sqrt{2}\epsilon_3 +  \epsilon_3^2) \right\}\\
\nonumber && \quad\;\; \leq \;\; \sum_{j=1}^p \P_q \left\{ \frac{1}{\nq} \left| \sum_{i=1}^{\nq} \bXtil_{q,i [j]} \Ztil_{q,i [j]} \right| > 2 \sigma_q \gamma_q (2\sqrt{2}\epsilon_3 +  \epsilon_3^2) \right\} \\
&& \quad\;\; \leq  \;\; 2p \hspace{0.015in}\exp\left( -\nq \epsilon_3^2\right) \; \equiv \; 2 \exp\left(-\nq \epsilon_3^2 + \log p\right), \label{P2:pf3_eqn4}
\end{eqnarray}
where the first inequality follows from an application of the union bound, and the next one follows from Lemma \ref{P2:lemma_3}. Thus, (\ref{P2:pf3_eqn4}) establishes the second bound in (\ref{P2:thm3_probbound_1}). \qed

To obtain the third bound in (\ref{P2:thm3_probbound_1}), we now note that:  $(\bXbar_{\nq} - \bmu_q) \sim \S\G_q(\sigma_q^2/\nq)$ and $\nq^{-1} \sum_{i=1}^{\nq}\Ztil_{q,i} \sim \S\G_q(\gamma_q^2/\nq)$, both of which follow from using Lemma \ref{P2:lemma_1} (ii). Hence, using Lemma \ref{P2:lemma_1} (v) and (i) respectively, and noting that $\E_q(\Ztil_q) = 0$, we have:
\begin{alignat}{2}
& \P_q \left( \left\|\bXbar_{\nq} - \bmu_q\right\|_{\infty}  > \epsilon_4 \right) \; && \leq  \; 2 \exp\left\{-\frac{\nq \epsilon_4^2}{2\sigma_q^2} + \log  p \right\}, \quad \mbox{and} \label{P2:pf3_eqn5}\\
& \P_q \left( \frac{1}{\nq} \left| \sum_{i=1}^{\nq} \Ztil_{q,i} \right| > \epsilon_5 \right) \; && \leq \; 2\exp \left( - \frac{\nq\epsilon_5^2}{2 \gamma_q^2} \right). \label{P2:pf3_eqn6}
\end{alignat}
From (\ref{P2:pf3_eqn5})-(\ref{P2:pf3_eqn6}), and  the definition of $\Tnqtt$ in (\ref{P2:tnq22_def}), we then have: $\forall \; \epsilon_4, \epsilon_5 > 0$,
\begin{align}
\P_q \left( \left\| \Tnqtt \right\|_{\infty} > \epsilon_4 \epsilon_5 \right) 
& \leq \hspace{0.015in} \P_q \left\{  \left\|\bXbar_{\nq} -\bmu_q \right\|_{\infty} \left( \frac{1}{\nq} \left| \sum_{i=1}^{\nq} \Ztil_{q, i} \right| \right)
> \epsilon_4 \epsilon_5 \right\} \nonumber \\
& \leq \hspace{0.002in} \P_q \left( \left\|\bXbar_{\nq} -\bmu_q \right\|_{\infty}  > \epsilon_4 \right)  +  \P_q \left(\frac{1}{\nq} \left| \sum_{i=1}^{\nq} \Ztil_{q,i} \right| > \epsilon_5 \right) \nonumber \\
& \leq \hspace{0.015in} 2 \exp\left(-\frac{ \nq \epsilon_4^2}{2\sigma_q^2} + \log p\right) + 2 \exp\left( - \frac{\nq \epsilon_5^2}{2 \gamma_q^2} \right). \label{P2:pf3_eqn7}
\end{align}
Thus, (\ref{P2:pf3_eqn7}) establishes the third bound in (\ref{P2:thm3_probbound_1}), and the proof of Theorem \ref{P2:THM_3} is now complete. We next complete the proof of Theorem \ref{P2:COR_1}, as a consequence of Theorem \ref{P2:THM_3}. \qed
\end{subsection}

\begin{proof}[Proof of Theorem \ref{P2:COR_1}]
The claim (\ref{P2:thm3_probbound_2}) in Theorem \ref{P2:COR_1} follows simply from noting the representations (\ref{P2:tnq_def})-(\ref{P2:tnq22_def}) of $\T_{\nq}$ in terms of $\Tnqo$, $\Tnqto$ and $\Tnqtt$, and straightforward usage of (\ref{P2:pf3_eqn3}), (\ref{P2:pf3_eqn4}) and (\ref{P2:pf3_eqn7}) through appropriate choices of $\{\epsilon_1, \hdots, \epsilon_5\}$ in terms of universal constants $(c_1,\hdots c_6) \equiv \bc$ as follows: $\epsilon_1 = \{ 2 \; \log (\nq^{c_1} p^{c_2})\}^{\half} \sigma_q $, for any $c_1, c_2 > 0$ such that $\mbox{max}\;(c_1, c_2) > 1$; $\epsilon_2 = \{c_3(1 - 2 \pi_q)/\nq\}^{\half}$ for any $c_3 >0$, and also noting the definition of $\piqtil$ when $\pi_q < 1/2$ (as assumed for this part); $\epsilon_3 = \{c_4 (\log  p)/\nq\}^{\half}$ for any $c_4 > 1$; and lastly, $\epsilon_4 = \{2 c_5 (\mbox{log} \; p)/\nq\}^{\half} \sigma_q$ and $\epsilon_5 = \{2 c_6 (\mbox{log} \; p)/\nq\}^{\half} \gamma_q$, for any $c_5 > 1$ and any $c_6 > 0$ respectively. This completes the proof of Theorem \ref{P2:COR_1}.
\end{proof}

\section{Illustrations for a Subclass of Models}\label{P2:splcase}

In this section, we illustrate how our main results can be applied for a special, but nevertheless interesting and frequently adopted, class of models. Apart from theoretical results and their proofs, the illustration also includes useful discussions and interpretations of these results under such settings.

We begin with a detailed result providing some key characterizations. In Theorem \ref{P2:thm_4} below (proved in Appendix \ref{P2:pf_thm4}), we characterize the distributional properties of $\{(S,\bX) \given$ $S \in \Isc_q\}$, and the behaviour of $\pi_q$,
for a specific family of distributions of $(Y,S,\bX')'$, where 
\begin{equation}
\bX \sim \Nsc_p(\bzero,\bSigma), \quad S = \balpha_0'\bX + \epsilon^* \;\; \mbox{and} \quad Y = 1(\bbeta_0'\bX + \epsilon > 0), 
\label{P2:splcase_setup}
\end{equation}
with
$\bSigma \succ 0$, $\epsilon^* \ind \bX$, $\epsilon \ind (S,\bX, \epsilon^*)$, $\epsilon^* \sim \Nsc_1(0,\sigma^2)$ for some $\sigma \geq 0$, and $\epsilon \sim \mbox{Logistic}(0,1)$. 
Thus, $S \given \bX$ and $Y \given \bX$ follow the standard linear and logistic regression models respectively, with $\sigma_S^2 \equiv \mbox{Var}(S) = \balpha_0' \bSigma \balpha_0 + \sigma^2$ and $\P(Y=1\given \bX, S) = \P(Y=1 \given \bX) = \psi(\bbeta_0'\bX)$, where $\psi(\cdot) = \exp(\cdot)/\{1 + \exp(\cdot)\}$ denotes the expit function. Further, note that \eqref{P2:splcase_setup} also implies
$$\bX \given S \sim  \Nsc_p(\bgamma_0, \Gamma) \;\; \mbox{with} \;\; \bgamma_0 = (\bSigma \balpha_0)/\sigma_S^2 \in \R^p, \;\; \Gamma_{p \times p} = (\bSigma - \sigma_S^2 \bgamma_0 \bgamma_0').$$
Let $\eta_0 \equiv \eta_0 (\bbeta_0) = (\bbeta_0' \bSigma \bbeta_0)^{1/2} \geq 0$ and $\rho_0 \equiv \rho_0(\balpha_0,\bbeta_0)  = \mbox{Corr}(\balpha_0'\bX,\bbeta_0'\bX)$. Assume (without any loss of generality) that $\rho \geq 0$ and let $\rhotil_0 \equiv \rhotil(\balpha_0,\bbeta_0) = \rho_0 (\balpha_0'\bSigma\balpha_0/\sigma_S^2)^{1/2} \geq 0$. Further, let $\lambda_{\max}(\bSigma) > 0$ denote the maximum eigenvalue of $\bSigma$.
%
Lastly, let $\Phi(\cdot)$ and $\phi(\cdot)$ denote the cumulative distribution function (CDF) and the density function of $\Nsc_1(0,1)$ respectively, and for any $q \in (0,1]$, let $z_q$ and $\zbar_q$ respectively denote the $(q/2)^{th}$ and $(1-q/2)^{th}$ quantiles of $\Nsc_1(0,1)$. Hence, $\delta_q = \sigma_Sz_q$ and $\deltabar_q = \sigma_S\zbar_q$.

\begin{theorem}\label{P2:thm_4}
Assume the set-up in (\ref{P2:splcase_setup}) above for $(Y,S,\bX')'$. Then,
\begin{enumerate}[(i)]
\item \label{P2:thm_4_r1}  $ \E_q(S) = 0$ and $\E_q(\bX) = \mathbf{0}$. Further, we have:
\begin{alignat*}{3}
& \nonumber \E(S \given S \geq \deltabar_q) \; && = \; - \hspace{0.015in} \E(S \given S \leq \delta_q) \; && = \; \sigma_S \frac{\phi(z_q)}{\Phi(z_q)}.  \\
& \nonumber \E(\bX \given S \geq \deltabar_q) \; && = \; - \hspace{0.015in} \E(\bX \given S \leq \delta_q) \; && = \; \bgamma_0 \frac{ \sigma_S\phi(z_q)}{\Phi(z_q)}.
\end{alignat*}
\item \label{P2:thm_4_r2} The $2^{nd}$ moments and the conditional variances, denoted  by $\Var_q(\cdot)$, of $S \given S \in \Isc_q$ and $\bX \given S \in \Isc_q$,  satisfy:
$\E(S^2 \given S \geq \deltabar_q) = \E(S^2 \given S \leq \delta_q) \equiv \Var_{q}(S)$ and $\E(\bX \bX' \given S \geq \deltabar_q) =  \E(\bX \bX' \given S \leq \delta_q) \equiv  \Var_q(\bX)$. Further 
\emph{
\begin{align*}
& \Var_{q}(S)  = \; \sigma_S^2 \left\{ 1 + \zbar_q \frac{\phi(z_q)}{\Phi(z_q)} \right\} \;\; \mbox{and}  \;\; \Var_q(\bX)    = \; \bSigma + \bgamma_0\bgamma_0'  \frac{ \sigma_S^2 \zbar_q \phi(z_q)}{\Phi(z_q)}.
\end{align*}
}
\item \label{P2:thm_4_r3} Let $\mbox{MGF}_{S,q}(t) \equiv \E_q\left\{\exp(tS)\right\}$ and $\mbox{MGF}_{\bX,q}(\bt) \equiv \E_q\left\{\exp\left(\bt'\bX\right)\right\}$ denote the moment generating functions (MGFs) of $S \given S \in \Isc_q$ and $\bX \given S \in \Isc_q$ respectively. Then, $\forall \; t \in \R$ and $\bt \in \R^p$,
\emph{
\begin{alignat*}{3}
& \mbox{MGF}_{S,q}(t) 
\; && = \; \frac{\exp\left(\sigma_S^2 t^2/2\right)}{2 \Phi(z_q)} \{ \Phi(z_q + \sigma_S t) + \Phi(z_q - \sigma_S t)\}, \;\; \mbox{and} \\ 
& \mbox{MGF}_{\bX,q}(\bt) 
\; && = \; \frac{\exp\left(\bt'\bSigma\bt/2\right)}{2 \Phi(z_q)} \{ \Phi(z_q + \sigma_S \bt'\bgamma_0) + \Phi(z_q - \sigma_S \bt'\bgamma_0)\}. 
\end{alignat*}
}
\item \label{P2:thm_4_r4} 
The MGFs for $S \given S \in \Isc_q$ and $\bX \given S \in \Isc_q$,
obtained in (iii) above, further satisfy the following subgaussian type bounds. For any $t \in \R$,
\emph{
\begin{alignat*}{3}
& \mbox{MGF}_{S,q}(t) \; && \leq \; \exp\left\{\half t^2 \sigma_S^2 (1+ 2 z_q^2)\right\}  \quad && \forall \; q \in (0,1/2], \\ 
& \mbox{MGF}_{S,q}(t) \; && \leq \; 4 \exp\left(\half t^2 \sigma_S^2\right)  \quad && \forall \; q \in (1/2, 1], 
\end{alignat*}
and for any $\bt \in \R^p$,
\begin{alignat*}{3}
& \mbox{MGF}_{\bX,q}(\bt) \; && \leq \; \exp\left[\half \|t\|_2^2 \left\{\lambda_{\max}(\bSigma) + 2 \sigma_S^2 z_q^2 \| \bgamma_0 \|_2^2 \right\}\right]   \quad
&& \forall \; q \in (0,1/2], \\ 
& \mbox{MGF}_{\bX,q}(\bt) \; && \leq \; 4 \exp\left\{\half \|t\|_2^2 \lambda_{\max}(\bSigma) \right\}  \quad
&&   \forall \; q \in (1/2,1]. 
\end{alignat*}
}
\item \label{P2:thm_4_r5} The misclassification error $\pi_q \equiv \half (\pi_q^+ + \pi_q^-)$ satisfies:  for any $q \in (0,1]$, 
\begin{eqnarray*}
\pi_q &\leq& \frac{\Phi\left(-\zbar_q - \sigma_S \bbeta_0'\bgamma_0 \right)}{\Phi\left(-\zbar_q \right)} \exp\left(\eta_0^2/2\right)   \\
&\leq& \exp\left\{\half \left(1 - \widetilde{\rho}_0^2 \right)\eta_0^2 - \zbar_q \widetilde{\rho}_0 \eta_0  \right\}  \frac {(\zbar_q^2 + 1)}{\zbar_q (\zbar_q + \widetilde{\rho}_0 \eta_0)}  \\
& \lesssim& C \exp\left\{\half \left(1 - \widetilde{\rho}_0^2 \right) \eta_0^2 - \zbar_q \widetilde{\rho}_0 \eta_0 \right\},
\end{eqnarray*}
where $C > 0$ is some universal constant.
\par\smallskip
\item \label{P2:thm_4_r6} The behavior of $\deltabar_q^2$ with respect to 
$q$ is reflected in the following bounds:
\begin{alignat*}{6}
&  \deltabar_q^2 \; && \equiv \; \sigma_S^2 \zbar_q^2 \; && = \; \sigma_S^2 z_q^2 \; && \equiv \; \delta_q^2 \; && \leq \; 2 \sigma_S^2 \log (q^{-1})      \qquad &&  \forall \; q \in (0,1].\\
& \deltabar_q^2 \; && \equiv \; \sigma_S^2 \zbar_q^2 \; && = \; \sigma_S^2 z_q^2 \; && \equiv \; \delta_q^2 \; && \geq \; 2 \sigma_S^2 \log \{(5q)^{-1}\}  \qquad && \forall \; q \in [0.0002,1].
\end{alignat*}
\end{enumerate}
\end{theorem}
\begin{remark}[Verification of Some of the Basic Conditions in our Main Results]\label{P2:splcase_rem1}
\emph{Result (iv) in Theorem \ref{P2:thm_4} implies that for any small $q$, $\bX \sim \S\G_q (\sigma_q^2)$ indeed for some $\sigma_q$, as required in Theorem \ref{P2:THM_3}. Further, result (vi) ensures that $\zbar_q^2$, and hence $\sigma_q^2$, diverges only at logarithmic orders, as $q \downarrow 0$, thereby showing that, at least in this case, $\sigma_q$ (and $\gamma_q$) appearing in the bound (\ref{P2:thm3_probbound_2}) only has a minor effect on the convergence rates of $\lambda \equiv 4 a_{\nq}$ and $\bbetahat_{\nq} (\lambda)$. Moreover, the strict positive definiteness of $\bSigma_q$, as shown by result (ii), combined with the fact that $\bX \sim \S\G_q (\sigma_q^2)$ also ensures that the restricted strong convexity in Assumption \ref{P2:strngconv_assmpn} can be ensured to hold in this case with high probability, using results from \citet{Negahban_2012} and \citet{Rud_Zhou_2013}, for some $\kappa_q \gtrsim \lambda_{\min}(\bSigma_q)$. As we show later in (\ref{P2:bqrate_eqn_new3}), $\lambda_{\min}(\bSigma_q) \geq \kappa$ $\forall \; q$ for some constant $\kappa > 0$ independent of $q$, and hence in this case, $\kappa_q$ further has no real impact on the rates of the bounds in our main results. 
}
\end{remark}

\begin{remark}[The Behavior of $\pi_q$ with Respect to $q$]\label{P2:splcase_rem2}\emph{
The bound for $\pi_q$ obtained in result (v) is a fairly sharp bound (especially for $q$ small enough). Apart from $\zbar_q$, the bound also depends critically on the constants $\eta_0$ and $\rhotil_0$ that can be respectively interpreted as the `strength' of the signal $\bbeta_0'\bX$, and (approximately) the correlation $\rho_0$ between the $\bbeta_0$ and the $\balpha_0$ directions. For this bound to obey a polynomial rate with respect to $q$, as we have assumed in our Extreme Tail Surrogacy Assumption \ref{P2:surr_assmpn}, 
it must behave as: $c_1 \exp (-c_2 \zbar_q^2 )$, owing to result (vi), for some constants $c_1, c_2 > 0$. However, treating $\eta_0$ and $\rhotil_0$ as universal constants, the bound clearly behaves only as: $c_1^* \exp(-c_2^* \zbar_q )$ for some constants $c_1^*, c_2^* >0$, thereby leading to a slower rate than desired.
}
\end{remark}

\begin{remark}[`Closeness' of $\bbeta_0$ and $\balpha_0$ for $\pi_q$ to be Desirably Small]\label{P2:splcase_rem3}\emph{
In the light of Remark \ref{P2:splcase_rem2}, it is 
helpful to envision a `regime' where $\eta_0$ and $\rhotil_0$ are allowed to vary with $q$, so that $\eta_0$ increases and $(1 -\rhotil_0^2)$ decreases, both albeit slowly enough, as $q \downarrow 0$. For a given choice of $q$, if $\eta_0$ is strong enough such that $\eta_0^2 \gtrsim d_1 \zbar_q^2$, and $(1 -\rhotil_0^2)$ is small enough so that $(1 -\rhotil_0^2)\eta_0^2 \lesssim d_2$, for some constants $d_1, d_2 > 0$, then clearly, the bound
for $\pi_q$ in result (v) starts behaving desirably as: $\pi_q \lesssim c_1 \exp (-c_2 \zbar_q^2)$ for some constants $c_1, c_2 > 0$. On the other hand, 
owing to result (vi), note that this regime also necessarily implies that:
$(1-\rho_0^2)$ $\leq $ $(1-\rhotil_0^2) \lesssim O(\zbar_q^{-2}) \lesssim O\{-(\log q)^{-1}\} \lesssim O\{1/(\log N)\}$, thus 
indicating that the $\balpha_0$ and $\bbeta_0$ directions must be fairly strongly correlated with each other, at least upto $1/(\log N)$ order. Hence, they must be `close' in this sense under this regime, a regime that is almost necessary in this case for 
(\ref{P2:misclass_error}) in Assumption \ref{P2:surr_assmpn}
to hold, and for our approach to be successful.
}
\end{remark}

Intuitively, the conclusion in Remark \ref{P2:splcase_rem3} 
makes sense since $\E(S \given \bX) \equiv \balpha_0'\bX$ and $\P(Y=1 \given\bX) \equiv \psi(\bbeta_0'\bX)$ are monotone in $\balpha_0'\bX$ and $\bbeta_0'\bX$ respectively, so that given $\bX$, the tails of $S$ will be closely linked to those of $\balpha_0'\bX$, and further owing to the surrogacy assumption in these tails, the corresponding $\psi(\bbeta_0'\bX)$ should be close to $0/1$ indicating that $\bbeta_0'\bX$ should also lie in its own tails, thus implying the necessity of a fairly strong correlation between the $\balpha_0$ and $\bbeta_0$ directions of $\bX$. These facts are illustrated in Figure \ref{P2:surrogacy_plot} (a)-(d).

However, this does \emph{not} trivialize the problem or our proposed methods. While the $\balpha_0$ direction does need to be `close' (in the sense of Remark \ref{P2:splcase_rem3}) to the $\bbeta_0$ direction in this case, it only needs to be so in $1/(\log N)$ order. Hence, while the $\balpha_0$ direction can indeed be near perfectly estimated from $\Usc_N^*$ at a rate of $O(N^{-1/2})$, it \emph{may not} be a reasonable estimator of the $\bbeta_0$ direction itself, which it might possibly be only able to estimate at sub-optimal logarithmic rates, instead of the polynomial rates we have ensured, under the additional structured sparsity assumptions, for our estimator based on $\Usc_{\nq}$.

In fact, this is \emph{exactly} what we observed in all our simulation studies in Section \ref{P2:sim} (and the real data analysis in Section \ref{P2:data_ex} as well), where the simulations were based on the set-up in (\ref{P2:splcase_setup}) with the $\balpha_0$ and $\bbeta_0$ directions being further `close' in the sense of Remark \ref{P2:splcase_rem3}. But the performance of $\balpha_0$ (estimated near perfectly from $\Usc_N^*$) as an estimator of the $\bbeta_0$ direction was clearly found to be much worse throughout, compared to the U$_{\mbox{\tiny LASSO}}$ as well as the S$_{\mbox{\tiny LASSO}}$ estimators, under all the settings and in terms of every criteria that we considered.

\begin{figure}
\centering
\mbox{
\subfigure[$\P(Y = 1 \given \bX)$ vs. $\bbeta_0'\bX$]{\includegraphics[scale=0.36]{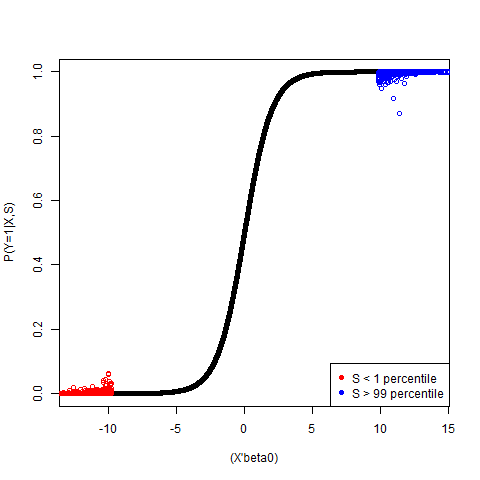}}
\subfigure[$S$ vs. $\balpha_0'\bX$]{\includegraphics[scale=0.36]{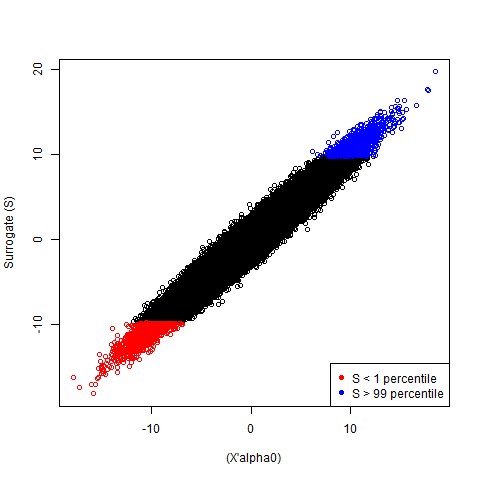}} } \\
\mbox{
\subfigure[$S$ vs. $\balpha_0'\bX$ (and $Y$)]{\includegraphics[scale=0.36]{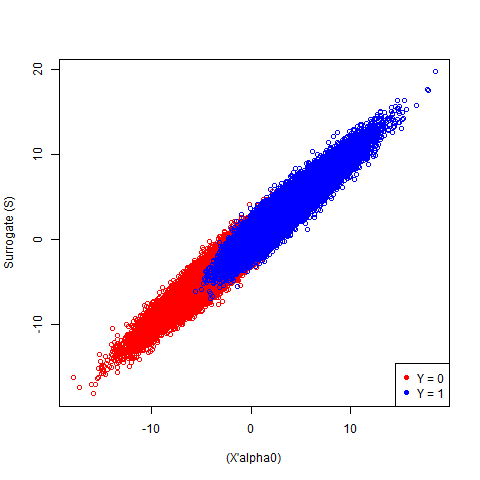}}
\subfigure[$\balpha_0'\bX$ vs. $\bbeta_0'\bX$]{\includegraphics[scale=0.36]{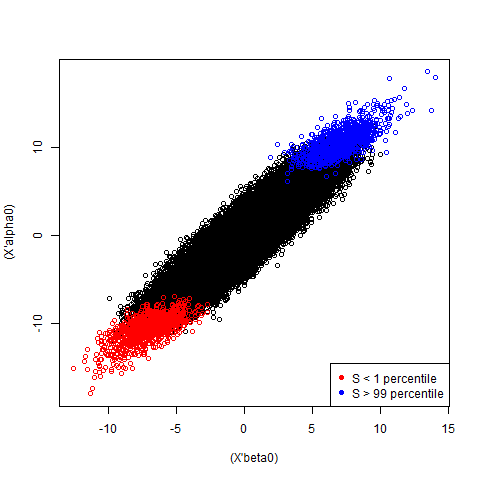}}
}
\caption{Graphical illustration of the behaviors of $\{S,\balpha_0'\bX,\bbeta_0'\bX, Y, \psi(\bbeta_0'\bX)\}$ with respect to each other, as well as the surrogacy of $S$ when $S \in \Isc_q$, under the set-up in (\ref{P2:splcase_setup}). The plots were generated using $N = 10^5$, $q = 0.02$, $p = 50$, $\bSigma = I_p$ (the $p \times p$ identity matrix), $\sigma = 1$, $\bbeta_0 = (\bone'_{c_p}, 0.5*\bone'_{c_p}, \bzero'_{p-2c_p})'$ with $c_p = \lfloor p^{\half} \rfloor$, and $\balpha_0 = \bbeta_0 + \bxi/(\log  N)$ with $\{ \bxi_{[j]}\}_{j=1}^p$ being $p$ \emph{fixed} realizations from $\mbox{Uniform}(0, 5)$. In plots (a), (b) and (d), the points with $S \leq \delta_q$ and $S \geq \deltabar_q$ are highlighted in red and blue respectively. In plot (c), the points with $Y = 0$ and $Y=1$ are highlighted in red and blue respectively. }\label{P2:surrogacy_plot}
\end{figure}

\begin{remark}[Recovery only upto Scalar Multiples \emph{even if} the Link is Known]\label{P2:splcase_rem4}
\emph{
It is worth understanding the `close' relationship of $Y \equiv 1(\bbeta_0'\bX + \epsilon > 0)$ in (\ref{P2:splcase_setup}) to its corresponding `link-free' version: $\mathbb{Y} \equiv 1(\bbeta_0'\bX > 0)$, when $S \in \Isc_q$ with $q$ small. Using Theorem \ref{P2:thm_4} and Lemma \ref{P2:lemma_4}, it can be shown that as $q \downarrow 0$, $\E_q\{(\bbeta_0'\bX)^2\} \asymp C_1 + C_2 \log(q^{-1})$ for some universal constants $C_1,C_2 > 0$ depending only on $(\bSigma,\balpha_0,\bbeta_0, \sigma_S)$ but not on $q$. Hence, the signal component in $Y$ (slowly) diverges, in $L_2(\P_q)$ norm, as $q \downarrow 0$. On the other hand, the noise component 
$\epsilon \ind (S,\bX)$ satisfies: $\E_q(\epsilon^2) = 1$ for \emph{any} $q$. Thus, given $S \in \Isc_q$ with $q$ small, the `strength' of the signal $\bbeta_0'\bX$ in $Y$ essentially `overwhelms' the \emph{fixed} level of noise contributed by $\epsilon$, and the latter plays very little role in determining the sign of $(\bbeta_0'\bX + \epsilon)$. Hence, 
the original outcome $Y$ quickly approaches the `noiseless' link-free variable $\mathbb{Y}$, that is $\P_q(Y \neq \mathbb{Y}) \approx 0$ for small enough $q$ (see also Figure \ref{P2:surrogacy_plot} for illustrations). It is important to note that $\mathbb{Y}$ is invariant to scalar transformations of $\bbeta_0$, thereby making $\bbeta_0$ identifiable \emph{only} upto scalar multiples. This, rather interestingly, indicates that under the surrogate aided unsupervised setting we have, whereby we need to move substantially into the tails of $S$, the knowledge of the link function for $Y$, even if available, virtually has no additional benefits and $\bbeta_0$ is (essentially) still identifiable only upto scalar multiples under $\P_q(\cdot)$ - something that our original SIM framework, with an unspecified link, guarantees anyway.
}

\emph{
Of course, our arguments here are somewhat heuristic and specific to the set-up in (\ref{P2:splcase_setup}). A rigorous and more generally applicable justification of this interesting (and counter-intuitive) fact needs further careful analyses beyond the scope of this work. Empirical evidence however seems to corroborate our intuition, at least under (\ref{P2:splcase_setup}), wherein we found that even if we use the knowledge of the link for $Y$ and fit a sparse logistic regression of $Y^*_q$ vs. $\bX$, it only tends to recover the $\bbeta_0$ upto a scalar multiple. Further, it is worth pointing out that an observation of a similar flavor has also been noted by \citet{Plan_Vershynin_2013_a}, albeit under a very different framework and associated assumptions, in the context of one-bit compressed sensing problems with `adversarial' bit flips (or corruptions).
}
\end{remark}

\begin{subsection}{Proof of Theorem \ref{P2:thm_4}}\label{P2:pf_thm4}

First, the results for $\{\E(S \given S \leq \delta_q)$, $\E(S^2 \given S \leq \delta_q)\}$ and $\{\E(S \given S \geq \deltabar_q), \; \E(S^2 \given S \geq \deltabar_q)\}$ follow directly from  Lemma \ref{P2:lemma_5} (i) with the choices of $a, b$ as: $\{a = - \infty, b = \delta_q\}$ and $\{ a = \deltabar_q, b = \infty\}$ respectively. Further, the results for $\E_q(S)$ and $\mbox{Var}_q(S)$ follow from noting that: $\E_q(S) = \{ \E(S \given S \leq \delta_q) + \E(S \given S \geq \deltabar_q)\} / 2$, and that: $ \{ \E(S^2 \given S \leq \delta_q) + \E(S^2 \given S \geq \deltabar_q)\} / 2$ $ = \E_q(S^2) \equiv \mbox{Var}_q (S)$ since $\E_q(S) = 0$. Next, for the corresponding results regarding $\bX$, we first note that under the assumed set-up, $\bX \given S$ follows a linear model given by: $\bX = \bgamma_0 S + \beps$ with $\beps \sim \Nsc_p(\bzero,\Gamma)$ and $\beps \ind S$, where $\bgamma_0$ and $\Gamma$ are as defined therein. Using this relation and the results already proved, the results for $\{\E(\bX \given S \leq \delta_q), \; \E(\bX \bX' \given S \leq \delta_q)\}$ and $\{\E(\bX \given S \geq \deltabar_q), \; \E(\bX \bX' \given S \geq \deltabar_q)\}$ now follow immediately. Further, the results for $\E_q(\bX)$ and $\mbox{Var}_q(\bX)$ follow from noting that: $\E_q(\bX) = \{ \E(\bX \given S \leq \delta_q) + \E(\bX \given S \geq \deltabar_q)\} / 2$, and that: $ \{ \E(\bX \bX' \given S \leq \delta_q) + \E(\bX \bX' \given S \geq \deltabar_q)\} / 2$ $ = \E_q(\bX \bX') \equiv \mbox{Var}_q (\bX)$ since $\E_q(\bX) = 0$. This completes the proof of all the results mentioned in (i) and (ii) in Theorem \ref{P2:thm_4}.  \qed

Next, we note that: $\forall \; t \in \R$, $\E_q\{\exp(tS)\} = [\E\{\exp(tS) \given S \leq \delta_q \} + \E\{\exp(tS) \given S \geq \deltabar_q\}]/2$, and the result in (iii) for $\mbox{MGF}_{S,q}(t)$ now follows directly from Lemma \ref{P2:lemma_5} (i) using the choices of $a, b$ as: $\{a = - \infty, b = \delta_q\}$ and $\{ a = \deltabar_q, b = \infty\}$ respectively. For $\mbox{MGF}_{\bX,q}(\bt)$, we note that: $\forall \; \bt \in \R^p$, $\E_q\left\{\exp\left(\bt'\bX\right)\right\} = \E_q[\exp\{(\bt'\bgamma_0)S + \bt'\beps\}]$ $= [\E\{\exp(\bt'\beps)\}] \mbox{MGF}_{S,q}(\bt'\bgamma_0)$, where in the last step, we also use $\beps \ind S$. The result now follows from using the result for $\mbox{MGF}_{S,q}(\cdot)$, and the standard expression for the MGF of $\beps \sim \Nsc_p(\bzero, \Gamma)$, as well as using the fact that $\Gamma = (\bSigma - \sigma_S^2 \bgamma_0 \bgamma_0')$. This completes the proof of all results in (iii). Further, all the bounds in (iv) for $\mbox{MGF}_{S,q}(\cdot)$ are straightforward implications of Lemma \ref{P2:lemma_5} (ii), and so are the bounds for $\mbox{MGF}_{\bX,q}(\cdot)$ in (iv), where we additionally use the standard inequalities: $\bt'\bSigma\bt \leq \lambda_{\max}(\bSigma) \| \bt \|_2^2$ and $ |\bt'\bgamma_0 |^2 \leq \| \bt \|_2^2 \| \bgamma_0 \|_2^2$ $\; \forall \; \bt \in \R^p$. This therefore completes the proof of all the MGF related results mentioned in (iii) and (iv) in Theorem \ref{P2:thm_4}.  \qed

Next, for the bounds on $\deltabar_q$ in result (vi), the upper bound is a straightforward consequence of the second inequality provided in Lemma \ref{P2:lemma_4}, and noting that: $q/2 = \Phibar(\zbar_q)$ and $\deltabar_q = \sigma_S \zbar_q$. The lower bound follows from the lower bound given in the first inequality in Lemma \ref{P2:lemma_4}. The restriction $q \geq 0.0002$ in the statement of the lower bound result in (vi) is needed to bound the quantity $\{(1 + \zbar_q^2)/\zbar_q\}$ that inevitably comes up while using the inequality from the lemma. In particular, this restriction implies that $\{(1 + \zbar_q^2)/\zbar_q\} \leq 5\sqrt{2/\pi}$ and ensures the final bound stated in the result. This completes the proof of (vi). \qed 

Finally, to show the results in (v), we first note that:
\begin{alignat*}{2}
& \pi_q^-   \; && = \; \E\{ \P(Y = 1 \given \bX, S) \given S \leq \delta_q\} \; \equiv \; \E\{ \psi(\bbeta_0'\bX) \given S \leq \delta_q \} \\
&  \; && \leq \; \E \left\{\exp\left(\bbeta_0'\bX\right)\given S \leq \delta_q \right\} \;\; \equiv \;\;  \frac{\Phi(-\zbar_q - \sigma_S \bbeta_0'\bgamma_0)}{\Phi(-\zbar_q)}\exp\left(\eta_0^2 / 2\right), \quad \mbox{and} \\
& \pi_q^+ \; && = \; \E\{ \P(Y = 0 \given \bX, S) \given S \geq \deltabar_q\} \; \equiv \; \E\{ \psi(-\bbeta_0'\bX) \given S \geq \deltabar_q \} \\
& \; && \leq \; \E \left\{\exp\left(-\bbeta_0'\bX \right) \given S \geq \deltabar_q \right\} \;\; \equiv \;\; \frac{\Phi(-\zbar_q - \sigma_S \bbeta_0'\bgamma_0)}{\Phi(-\zbar_q)} \exp\left(\eta_0^2/2\right).
\end{alignat*}
Therefore,
\begin{equation}
\pi_q  \; \equiv \; \half(\pi_q^- + \pi_q^+) \; \leq \; \frac{\Phi(-\zbar_q - \sigma_S \bbeta_0'\bgamma_0)}{\Phi(-\zbar_q)}\exp\left(\eta_0^2/2\right). \label{P2:pf4_eqn3}
\end{equation}
The final bounds above follow, similar to the earlier proofs for the results in (iii), from straightforward uses of the results regarding MGFs of truncated normal distributions given in Lemma \ref{P2:lemma_5}, as well as use of the relationship between $\bX$ and $S$ given by: $\bX = \bgamma_0 S + \beps$ with $\beps \sim \Nsc_p(\bzero,\Gamma)$ and $\beps \ind S$, and noting the definitions of $\bgamma_0$, $\Gamma$ and $\eta_0$. (\ref{P2:pf4_eqn3}) therefore establishes the first bound in result (v) of Theorem \ref{P2:thm_4}. The subsequent bounds in result (v) now follow from (\ref{P2:pf4_eqn3}) along with straightforward uses of the first inequality (both the upper and lower bounds) given in Lemma \ref{P2:lemma_4} noting that $\bbeta_0'\bgamma_0 > 0$ by assumption, as well as using the fact that $\sigma_s \bbeta_0'\bgamma_0 \equiv \rhotil_0 \eta_0$. All the claims in result (v) are now established, and the proof of Theorem \ref{P2:thm_4} is complete. \qed
\end{subsection}

\subsection{Bounds on the Behavior of the Scale Multiplier} \label{P2:bqrate}
In this section, we investigate the behavior of the constant $b_q$ that appears in the representation (\ref{P2:thm1_eqn1}) of $\bbetabar_q$ in Theorem \ref{P2:THM_1}, as well as in the deviation bound (\ref{P2:thm2_dev_bound}) in Theorem \ref{P2:THM_2}. It essentially represents the scalar multiple upto which our $\ULASSO$ estimator can recover the $\bbeta_0$ direction based on $\Usc_{\nq}$, for any given $q$. Hence, if one is interested in the convergence rates of the normalized $\ULASSO$ estimator as an estimator of the normalized $\bbeta_0$ direction: $\bbeta_0/ \| \bbeta_0\|_2$ (identifiable upto a sign), then the behavior of the constant $b_q$ with respect to $q$ becomes of interest, as $|b_q|^{-1}$ will now appear in the deviation bound for the normalized estimator, and therefore if $b_q$ does converge to $0$ as $q \downarrow 0$, its rate of convergence will affect the final rate of the deviation bound. In this section, we obtain bounds (shown in (\ref{P2:bqrate_finalrate}) later) for $b_q$, under the set-up in (\ref{P2:splcase_setup}), and show therein that while $b_q$ converges to $0$, it only does so at a very slow rate (square root of logarithmic orders) and hence, has no real impact on the (dominating) polynomial part of the rate in the deviation bound in this case.

We assume the set-up in (\ref{P2:splcase_setup}) for $(S,\bX')'$ and adopt all associated notations and quantities introduced in Appendix \ref{P2:splcase}. However, unlike (\ref{P2:splcase_setup}), we do not make any specific assumptions on $Y$, and only require (\ref{P2:misclass_error}) to hold, so that for all small enough $q$, $\pi_q \lesssim q^{\nu}$ for some $\nu > 0$. 
Throughout in this analysis, we use $\{c_k\}_{k\geq 1}$ to denote generic universal positive constants that do not depend on $q$ or any distributional parameter of $\P_q(\cdot)$, but may involve certain distributional parameters of $\P(\cdot)$ such as 
$(\bSigma,\balpha_0,\bbeta_0,\rho_0,\sigma_S,p)$. Since this is essentially a `population' level analysis and we are mostly interested in
behavioral changes with respect to $q$, such constants will be treated as fixed. Lastly, for any symmetric matrix $\mathbf{A} \equiv [a_{ij}]_{p \times p}$, let 
$\| \mathbf{A} \|_1 \equiv \underset{1 \leq j \leq p}{\max} \sum_{i=1}^p |a_{ij}|$ denote the matrix $L_1$ norm of $\mathbf{A}$.

To analyze $|b_q|$, we first recall from Theorem \ref{P2:THM_1}, with all notations as defined therein, the following representations of $\bbetabar_q$ and $\balphabar_q$ for any $q \in (0,1]$,
\begin{equation}
\bbetabar_q \; = \; b_q \bbeta_0 + a_q \balpha_0, \;\;\; \mbox{and} \;\;\; \balphabar_q \; = \; a_q^*\balpha_0.\label{P2:bqrate_eqn0}
\end{equation}
Further, the definitions (\ref{P2:y_sqloss})-(\ref{P2:s_sqloss}) of $\bbetabar_q$ and $\balphabar_q$ also imply that for any $q$,
\begin{equation}
\bbetabar_q \;=\; \bSigma_q^{-1} \E_q(\bX Y), \;\;\; \mbox{and} \;\;\;  \balphabar_q \;= \; \bSigma_q^{-1} \E_q(\bX Y^*_q),
\label{P2:bqrate_eqn1}
\end{equation}
where we also used the facts that $\bSigma_q \succ 0$ and $\E_q(\bX) = \bzero$ for any $q$, both of which follow from Theorem \ref{P2:thm_4}.
Now, while the behavior of $\bbetabar_q$ and $b_q$ is still somewhat difficult to analyze directly, the behavior of $\balphabar_q$ is relatively easier to analyze under the assumed set-up for $(S,\bX')'$. We will therefore take an indirect route, where we first evaluate $\balphabar_q$ explicitly under the assumed set-up, then bound $\bbetabar_q$ and $b_q$ in terms of $\balphabar_q$ and constants, 
and finally use the exact form of $\balphabar_q$ to obtain precise bounds on the behavior of $\bbetabar_q$ and $b_q$. 

To this end, we first note that, using (\ref{P2:bqrate_eqn1}), the following \emph{general} bound holds:
\begin{eqnarray}
\nonumber \left\| \bbetabar_q - \balphabar_q \right\|_{\infty} &\equiv& \left\| \bSigma_q^{-1}\left[ \E_q \{\bX(Y - Y^*_q)\}\right] \right\|_{\infty} \\
\nonumber &\leq& \left\| \bSigma_q^{-1} \right\|_1 \left\| \E_q \{\bX(Y - Y^*_q)\}\right\|_{\infty} \\
&\leq& \left\| \bSigma_q^{-1} \right\|_1 \left[\E_q\left\{(Y-Y^*_q)^2\right\}\right]^{\half}  \left[\underset{1 \leq j \leq p}{\max} \left\{\E_q\left(\bX_{[j]}^2\right)\right\}^{\half}\right], \label{P2:bqrate_eqn2} 
\end{eqnarray}
where the bounds follow from standard matrix norm inequalities and the Cauchy-Schwartz inequality. Note that $\E_q\{(Y - Y^*_q)^2\} = \pi_q$, and since $\bX \sim \S\G_q(\sigma_q^2)$ due to Theorem \ref{P2:thm_4}, $\E_q(\bX^2_{[j]}) \lesssim \sigma_q^2$ $\; \forall \; 1 \leq j \leq p$. Further, 
$\pi_q \lesssim q^{\nu}$ due to (\ref{P2:misclass_error}), and $\sigma_q^2 \lesssim \log(q^{-1})$ due to Theorem \ref{P2:thm_4}. 
Using these facts in (\ref{P2:bqrate_eqn2}), we then have:
\begin{equation*}
\nonumber \left\| \E_q\{ \bX(Y - Y^*_q)\} \right\|_{\infty} \; \lesssim \; c_1 \left(\pi_q \sigma_q^2\right)^{\half}
  \; \lesssim \; c_2 \left\{ q^{\nu} \log (q^{-1}) \right\}^{\half},
\end{equation*}
and therefore,
\begin{equation}
\left\| \bbetabar_q - \balphabar_q \right\|_{\infty} \; \lesssim \; c_1 \left\| \bSigma_q^{-1} \right\|_1 \left(\pi_q \sigma_q^2\right)^{\half}
  \; \lesssim \; c_2 \left\| \bSigma_q^{-1} \right\|_1 \left\{ q^{\nu} \log (q^{-1}) \right\}^{\half}. \label{P2:bqrate_eqn3}
\end{equation}
For (\ref{P2:bqrate_eqn3}), stronger bounds involving $\pi_q$ instead of $\pi_q^{\half}$ can also be obtained if we assume $\bX$ is bounded or more generally, $\E_q(\bX \given Y \neq Y^*_q)$ is bounded. However, the bound in (\ref{P2:bqrate_eqn3}), obtained under weaker conditions, is still polynomial in $q$, which suffices for all our purposes here.
Now, using (\ref{P2:bqrate_eqn0}) and (\ref{P2:bqrate_eqn3}), we have:
\begin{eqnarray*}
\nonumber |b_q| \left\| \bbeta_0 \right\|_{\infty} &\equiv& \left\| \bbetabar_q - a_q \balpha_0 \right\|_{\infty} \;\; \leq \;\; \left\| \bbetabar_q \right\|_{\infty} + |a_q| \left\| \balpha_0\right\|_{\infty} \\
\nonumber &\lesssim& \left\| \balphabar_q \right\|_{\infty} + |a_q| \left\| \balpha_0\right\|_{\infty} + c_2 \left\| \bSigma_q^{-1} \right\|_1 \left\{ q^{\nu} \log (q^{-1}) \right\}^{\half}, \quad \mbox{and}\\
\nonumber |b_q| \left\| \bbeta_0 \right\|_{\infty} &\equiv& \left\| \bbetabar_q - a_q \balpha_0 \right\|_{\infty} \;\; \geq \;\; \left\| \bbetabar_q \right\|_{\infty} - |a_q| \left\| \balpha_0\right\|_{\infty} \\
\nonumber &\gtrsim& \left\| \balphabar_q \right\|_{\infty} - |a_q| \left\| \balpha_0\right\|_{\infty} - c_2 \left\| \bSigma_q^{-1} \right\|_1 \left\{ q^{\nu} \log (q^{-1}) \right\}^{\half}.
\end{eqnarray*}
Hence, for all small enough $q \in (0,1]$ of interest, we have:
\begin{eqnarray}
|b_q| &\lesssim& c_3 \left\| \balphabar_q \right\|_{\infty} + c_4 |a_q| + c_5 \left\| \bSigma_q^{-1} \right\|_1 \left\{ q^{\nu} \log (q^{-1}) \right\}^{\half}, \;\; \mbox{and}  \label{P2:bqrate_eqn4} \\
|b_q| &\gtrsim& c_3 \left\| \balphabar_q \right\|_{\infty} - c_4 |a_q| - c_5 \left\| \bSigma_q^{-1} \right\|_1 \left\{ q^{\nu} \log (q^{-1}) \right\}^{\half},  \label{P2:bqrate_eqn5}
\end{eqnarray}
for some constants $(c_3,c_4,c_5) > 0$ that do not depend on $q$. (\ref{P2:bqrate_eqn4})-(\ref{P2:bqrate_eqn5}) therefore provide upper and lower bounds for $|b_q|$ and show that in order to analyze the behavior of $b_q$ as $q$ decreases, it suffices to analyze those of $\|\balphabar_q\|_{\infty}$ and $|a_q|$.

Towards this goal, we next explicitly 
evaluate $\bSigma_q^{-1}$ and $\balphabar_q$, which is where we make maximum use of the assumed set-up for $(S,\bX')'$. Using the results of Theorem \ref{P2:thm_4} and Woodbury's formula for matrix inverses, we have: for each $q \in (0,1]$,
\begin{eqnarray}
\nonumber \bSigma_q^{-1} &\equiv& \left\{ \bSigma + \left(\bgamma_0\bgamma_0' \right) \frac{\sigma_S^2 \zbar_q \phi(\zbar_q)}{\Phibar(\zbar_q)}\right\} ^{-1} \\ 
\nonumber &=& \left\{ \bSigma + \bSigma \left(\frac{\balpha_0\balpha_0'}{\sigma_S^2}\right) \bSigma \xi_q \right\}^{-1}, \quad \mbox{where} \;\; \xi_q \; \overset{\mbox{\tiny{def}}}{=} \; \zbar_q \frac{\phi(\zbar_q)}{\Phibar(\zbar_q)} \;\;\; \forall \; q, \\
\nonumber &=& \bSigma^{-1} - \bSigma^{-1} \bSigma \frac{(\balpha_0 \balpha_0')/\sigma_S^2}{\{1 + (\balpha_0'\bSigma \bSigma^{-1} \bSigma \balpha_0)(\xi_q/\sigma_S^2)\}} \bSigma \bSigma^{-1} \xi_q\\
&=& \bSigma^{-1} - \widetilde{\xi}_q(\balpha_0 \balpha_0'), \quad \mbox{where} \;\; \widetilde{\xi}_q \; \overset{\mbox{\tiny{def}}}{=} \; \frac{\xi_q}{\sigma_S^2 + \xi_q(\balpha_0'\bSigma \balpha_0)} \;\;\; \forall \; q. \label{P2:bqrate_eqn6}
\end{eqnarray}
%
Next, using the definition of $\balphabar_q$, as well as (\ref{P2:bqrate_eqn6}) above, we have: $\forall \; q \in (0,1]$,
\begin{eqnarray}
\nonumber \balphabar_q &\equiv& \bSigma_q^{-1} \E_q(\bX Y^*_q) \; = \; \bSigma_q^{-1} \left\{ \E_q(\bX \given Y^*_q = 1) \; \P_q(Y^*_q =1)\right\} \\
\nonumber &=& \frac{1}{2}\bSigma_q^{-1} \E(\bX \given S \geq \deltabar_q) \; = \;  \frac{1}{2}\bSigma_q^{-1} \bgamma_0 \left\{ \sigma_S \frac{\phi(\zbar_q)}{\Phibar(\zbar_q)} \right\} \; \mbox{[using Theorem \ref{P2:thm_4}],}\\
\nonumber &=& \frac{1}{2}\bSigma_q^{-1} \frac{(\bSigma \balpha_0)}{\sigma_S^2} \left(  \sigma_S \frac{\xi_q}{\zbar_q} \right) \qquad [\mbox{using the defintions of} \; \bgamma_0 \; \mbox{and} \; \xi_q],\\
\nonumber &=& \frac{1}{2} \left\{ \bSigma^{-1} - \frac{\xi_q (\balpha_0 \balpha_0')}{\sigma_S^2 + \xi_q(\balpha_0'\bSigma \balpha_0)} \right\} \left(\bSigma \balpha_0\right) \left( \frac{\xi_q}{\sigma_S \zbar_q} \right) \\
\nonumber &=&  \frac{1}{2} \balpha_0  \left\{ 1 - \frac{\xi_q (\balpha_0'\Sigma \balpha_0)}{\sigma_S^2 + \xi_q(\balpha_0'\Sigma \balpha_0)} \right\} \left( \frac{\xi_q}{\sigma_S \zbar_q} \right) \\
\nonumber &=& \frac{1}{2} \balpha_0 \left\{ \frac{\sigma_S^2}{\sigma_S^2 + \xi_q(\balpha_0'\bSigma \balpha_0)} \right\} \left( \frac{\xi_q}{\sigma_S \zbar_q} \right) \\
 &=& \xi^*_q \balpha_0, \quad \mbox{where} \;\; \xi^*_q \; \overset{\mbox{\tiny{def}}}{=} \; \frac{\sigma_S \xi_q}{2 \zbar_q \{\sigma_S^2 + \xi_q(\balpha_0'\Sigma \balpha_0)\}} \; \equiv \; \frac{\sigma_S \widetilde{\xi}_q}{2 \zbar_q} \;\;\; \forall \; q.  \label{P2:bqrate_eqn7}
\end{eqnarray}
Further, owing to Lemma \ref{P2:lemma_4} and Theorem \ref{P2:thm_4}, note that $\xi_q \geq 0$ and $\zbar_q \geq 0$ also satisfy the following bounds: for any $q \in [0.0002,1]$,
\begin{equation*}
\zbar_q^2 \; \leq \; \xi_q \; \leq \; (1+\zbar_q^2) \quad \mbox{and} \quad 2 \log(0.2 q^{-1}) \; \leq \; \zbar_q^2 \; \leq \; 2 \log(q^{-1}),
\end{equation*}
and therefore, $\xi_q$, $\widetilde{\xi}_q$ and $\xi^*_q$ satisfy the following bounds: $\forall \; q \in [0.0002,1]$ and for some constants $\{c_k\}_{k=6}^{11} > 0$ that do not depend on $q$,
\begin{eqnarray}
&& c_6 \log(q^{-1}) \; \lesssim \; \xi_q \; \lesssim \; c_7 \log(q^{-1}), \quad c_8 \; \lesssim \; \widetilde{\xi}_q \; \lesssim c_9, \quad \mbox{and} \label{P2:bqrate_eqn8} \\
&& \frac{c_{10}}{\{\log(q^{-1})\}^{\half}} \; \lesssim \; \xi^*_q \; \lesssim \; \frac{c_{11}}{\{\log(q^{-1})\}^{\half}}.  \label{P2:bqrate_eqn9}
\end{eqnarray}
Using (\ref{P2:bqrate_eqn9}) in (\ref{P2:bqrate_eqn7}), we then have:
\begin{equation}
\frac{c_{12}}{\{\log(q^{-1})\}^{\half}} \; \lesssim \; \left\| \balphabar_q \right\|_{\infty} \; \lesssim \; \frac{c_{13}}{\{\log(q^{-1})\}^{\half}}.  \label{P2:bqrate_eqn_new1}
\end{equation}
(\ref{P2:bqrate_eqn_new1}) therefore establishes the behavior of $\left\| \balphabar_q \right\|_{\infty}$ and shows that it converges to $0$ as $q \downarrow 0$, but only at a very slow rate of $(- \log q)^{- \half}$. Recall that $\left\| \balphabar_q \right\|_{\infty}$ was one of the key quantities in the bounds (\ref{P2:bqrate_eqn4})-(\ref{P2:bqrate_eqn5}). We next focus on the behaviors of $\|\bSigma_q^{-1}\|_1$ and $|a_q|$, the other key quantities in (\ref{P2:bqrate_eqn4})-(\ref{P2:bqrate_eqn5}) that depend on $q$. To this end, we first note that for each $q$, $\bSigma_q \succ 0$ owing to Theorem \ref{P2:thm_4} and hence, $\bSigma_q^{-1} \succ 0$ as well. Further, $\bSigma^{-1}_q$ depends on $q$ only through the quantity $\widetilde{\xi}_q$. Using (\ref{P2:bqrate_eqn6}) and (\ref{P2:bqrate_eqn8}), we then have: $\forall \; q \in [0.0002,1]$,
\def\tr{\mbox{trace}}
\begin{equation}
\| \bSigma_q^{-1} \|_1 \; \leq \; \|\bSigma^{-1}\|_1 + \widetilde{\xi}_q \| \balpha_0 \balpha_0'\|_1 \; \lesssim \; c_{14}, \label{P2:bqrate_eqn_new2}
\end{equation}
and further,
\begin{eqnarray}
\nonumber \left\{\lambda_{\min}(\bSigma_q)\right\}^{-1} &=& \lambda_{\max}(\bSigma_q^{-1}) \; \leq \; \tr(\bSigma_q^{-1}) \;\; \equiv \;\; \tr(\bSigma^{-1}) - \widetilde{\xi}_q \| \balpha_0\|_2^2 \\
&\leq& \tr(\bSigma^{-1}) + \widetilde{\xi}_q \| \balpha_0\|_2^2 \;\; \lesssim \;\; c_{15} \label{P2:bqrate_eqn_new3}, 
\end{eqnarray}
for some constants $c_{14}, c_{15} > 0$ independent of $q$, where for any matrix $\mathbf{A} \succ \bzero$, $\lambda_{\min}(\mathbf{A}) > 0$ and $\lambda_{\max}(\mathbf{A}) > 0$ respectively denote its minimum and maximum eigenvalues.

(\ref{P2:bqrate_eqn_new3}) shows that $\lambda_{\min}(\bSigma_q)$ is bounded away from $0$ uniformly in $q$. This is useful since with $\bX \sim \S\G_q(\sigma_q^2)$ and $\bSigma_q \succ 0$, it can be already shown, using results from \citet{Negahban_2012} and \citet{Rud_Zhou_2013}, that Assumption \ref{P2:strngconv_assmpn} holds with high probability by choosing the constant $\kappa_q$ therein as: $\kappa_q \asymp \lambda_{\min}(\bSigma_q)$, and (\ref{P2:bqrate_eqn_new3}) therefore now ensures that: $\kappa_q \geq \kappa > 0$, for some universal constant $\kappa$ independent of $q$.

Now, as shown in the proof of Theorem \ref{P2:THM_2}, for any choice of the tuning parameter $\lambda$ that is $(\bbeta_0,\balpha_0,q)$-admissible, as in Assumption \ref{P2:sparsity_assmpn}, $|a_q| \lesssim \lambda/\kappa_q$. Hence, due to (\ref{P2:bqrate_eqn_new3}) and the subsequent remarks above, $|a_q| \lesssim \lambda$ for any such $\lambda$. Therefore, as long as for some $\theta > 0$, $\exists \; \lambda \asymp q^{\theta}$ that is $(\bbeta_0,\balpha_0,q)$-admissible (note that the choice of $\lambda \equiv 4 a_{\nq}$, considered throughout for our main results in Section \ref{P2:ULASSO}, is a special case of such a $\lambda$), we have: $|a_q| \lesssim q^{\theta}$. Using this, as well as (\ref{P2:bqrate_eqn_new1}) and (\ref{P2:bqrate_eqn_new2}), in the original bounds (\ref{P2:bqrate_eqn4})-(\ref{P2:bqrate_eqn5}), 
we have: 
\begin{equation}
\left| |b_q|  - \frac{c^*}{\{\log(q^{-1})\}^{\half}} \right| \; \lesssim  \; d^* q^{\nu^*}\{\log(q^{-1})\}^{\half} \qquad \forall \; q \;\; \mbox{small enough}, \label{P2:bqrate_finalrate}
\end{equation}
where $\nu^* \equiv \min (\nu/2, \theta)$  and $c^*, d^* > 0$ are constants independent of $q$.

(\ref{P2:bqrate_finalrate}) therefore establishes the behavior of $|b_q|$ and shows that if $q$ goes to $0$, $|b_q|$ does converge to $0$, but only at a slow rate of $O(1/\sqrt{- \log q})$. In particular, if $q \asymp N^{-\eta}$ for some $0 < \eta < 1$, $|b_q|$ converges to $0$ at a slow rate of $O(1/\sqrt{\log N})$ and consequently, $|b_q^{-1}|$ diverges at a rate of $O(\sqrt{\log N})$. Hence, at least under the set-up for $(S,\bX')'$ considered here, the rate of divergence of the constant $|b_q|^{-1}$ that inevitably appears in the deviation bound for the normalized $\ULASSO$ estimator, 
has no real impact on the bound (apart from the minor effect of slowing down the overall rate by a $\sqrt{\log N}$ factor), and the (dominating) polynomial part of the rate in our bound stays unaffected. \qed

\subsection{Closed Form Expression of the Scale Multiplier} 
For any given $q$, an explicit characterization of the constant $b_q$ can be useful, for instance, to check if $b_q$ is non-zero (and hence, ensure that we are not ending up estimating the `zero'-direction rather than the $\bbeta_0$ direction). To characterize $b_q$, we shall more generally compute the constant $b_{\bv}$ in Assumption \ref{P2:dlc_assmpn} for any $\bv \in \R^p$, and then simply observe that $b_q$ is given by $b_{\bv}$ for $\bv = \bbetabar_q$ as in (\ref{P2:y_sqloss}).

To this end, we note that for any $\bv \in \R^p$, the constants $b_{\bv}$ and $a_{\bv}$ in Assumption \ref{P2:dlc_assmpn} must be the regression coefficients obtained from a (population based) least squares regression of $\bv'\bX$ on $(\bbeta_0'\bX,\balpha_0'\bX)$. Using standard theory of least squares estimation and some straightforward algebra,
it is then easy to show that: 
for any $\bv \in \R^p$, 
\begin{eqnarray}
\nonumber b_{\bv} &=& \frac{1}{(1 - \rho^2)(\bbeta_0'\bSigma \bbeta_0)^{\half}} \left\{ \frac{\bv'\bSigma\bbeta_0}{(\bbeta_0'\bSigma \bbeta_0)^{\half}} - \rho \frac{\bv'\bSigma\balpha_0}{(\balpha_0'\bSigma \balpha_0)^{\half}}\right\} \\
&=& \frac{1}{(1 - \rho^2)(\bbeta_0'\bSigma \bbeta_0)^{\half}} \; \mbox{Cov}\left[\bv'\bX,  \left\{\frac{\bbeta_0'\bX}{(\bbeta_0'\bSigma \bbeta_0)^{\half}} - \rho \frac{\balpha_0'\bX}{(\balpha_0'\bSigma \balpha_0)^{\half}} \right\} \right], \label{P2:bqrate_eqn10} 
\end{eqnarray}
where $\rho = \mbox{Corr}(\bbeta_0'\bX, \balpha_0'\bX)$. Hence, $b_q$ is given by (\ref{P2:bqrate_eqn10}) above with the choice: $\bv = \bbetabar_q$. \qed 

\phantomsection
\addcontentsline{toc}{section}{References}
\bibliographystyle{imsart-nameyear}
\bibliography{P2_Arxiv_R2_Biblio}

\end{document}